\documentclass[12pt]{article}

\usepackage{booktabs}

\usepackage[utf8]{inputenc}

\usepackage{amsmath}

\usepackage{amsthm}
\usepackage{geometry}
\usepackage{setspace}

\newtheorem{theorem}{Theorem}

\newtheorem{corollary}{Corollary}
\newtheorem{definition}{Definition}
\newtheorem{example}{Example}
\newtheorem{lemma}{Lemma}
\newtheorem{proposition}{Proposition}
\newtheorem{remark}{Remark}

\usepackage{natbib}
\usepackage[dvips]{graphics}

\usepackage{mathpazo}
\usepackage{enumerate}

\usepackage{amsfonts}
\usepackage{amssymb}

\usepackage{mathrsfs}

\usepackage[usenames]{color}
\usepackage{cancel}

\usepackage{pgf,tikz}

\usetikzlibrary{decorations.pathreplacing,angles,quotes}
\usetikzlibrary{arrows.meta}
\usetikzlibrary{arrows, decorations.markings}
\tikzset{myptr/.style={decoration={markings,mark=at position 1 with %
       {\arrow[scale=2,>=stealth]{>}}},postaction={decorate}}}

\usepackage{hyperref}
\hypersetup{
    colorlinks,
    citecolor=blue,
    filecolor=black,
    linkcolor=blue,
    urlcolor=blue
}

\usetikzlibrary{arrows}

\usepackage{fancyhdr}

\usepackage{soul}

\usepackage{emptypage}

\newcommand*\samethanks[1][\value{footnote}]{\footnotemark[#1]}

\usepackage[T1]{fontenc}
\DeclareFontFamily{T1}{calligra}{}
\DeclareFontShape{T1}{calligra}{m}{n}{<->s*[1.44]callig15}{}
\DeclareMathAlphabet\mathcalligra   {T1}{calligra} {m} {n}

\usepackage{multirow}
\usepackage{array}
\usepackage{mathtools}
\usepackage{arydshln}
\usepackage{threeparttable}
\usepackage{booktabs,caption}

\usepackage{ifthen}
\newboolean{showcomments}
\setboolean{showcomments}{true}

\newcommand{\pablo}[1]{  \ifthenelse{\boolean{showcomments}}
{\textcolor{green!50!black}{(T: #1)}}{}}
\newcommand{\marcelo}[1]{\ifthenelse{\boolean{showcomments}}
{\textcolor{red}{(M: #1)}}{}}
\newcommand{\agustin}[1]{  \ifthenelse{\boolean{showcomments}}
{\textcolor{blue!50!black}{(T: #1)}}{}}

\begin{document}

\title{Regret-free truth-telling voting rules%
\thanks{%
We are grateful to Jordi Massó for valuable comments and suggestions.
 We acknowledge financial support
from UNSL through grants 032016, 030120, and 030320, from Consejo Nacional
de Investigaciones Cient\'{\i}ficas y T\'{e}cnicas (CONICET) through grant
PIP 112-200801-00655, and from Agencia Nacional de Promoción Cient\'ifica y Tecnológica through grant PICT 2017-2355.}}

\author{R. Pablo Arribillaga\thanks{
Instituto de Matem\'{a}tica Aplicada San Luis, Universidad Nacional de San
Luis and CONICET, San Luis, Argentina, and RedNIE. Emails: \href{mailto:rarribi@unsl.edu.ar}{rarribi@unsl.edu.ar} 
and \href{mailto:abonifacio@unsl.edu.ar}{abonifacio@unsl.edu.ar} 
} \and Agustín G. Bonifacio\samethanks[2] \and 
Marcelo A. Fernandez\thanks{%
Department of Economics, University of Nevada, Reno, USA. Email: \href{marceloarielf@unr.edu}{marceloarielf@unr.edu}}
}

\date{\today}

\maketitle

\begin{abstract}

We study the ability of different classes of voting rules to induce agents to report their preferences truthfully,  if agents want to avoid regret. 
First, we show that  regret-free truth-telling is equivalent to strategy-proofness among tops-only rules. Then, we focus on three important families of (non-tops-only) voting methods: maxmin, scoring, and Condorcet consistent ones. We prove positive and negative results for both neutral and anonymous versions of maxmin and scoring rules. In several instances we provide necessary and sufficient conditions. We also show that Condorcet consistent rules that satisfy a mild monotonicity requirement are not regret-free truth-telling. Successive elimination rules fail to be regret-free truth-telling despite not satisfying the monotonicity condition.   Lastly, we provide two characterizations for the case of three alternatives and two agents.

\bigskip

\bigskip


\end{abstract}

\section{Introduction}

Voting rules are procedures that allow a group of agents to select an alternative, among many, according to their preferences.
Naturally, their vulnerability to manipulation is a primary concern.
Thus, it is desirable that voting rules are strategy-proof, meaning that it is always in the best interest of agents to report their preferences truthfully regardless of the behavior of others. 

Unfortunately, outside of a dictatorship, there is no strategy-proof voting rule when more than two alternatives, and all possible preferences over alternatives, are considered \citep{gibbard1973manipulation,satterthwaite1975strategy}.
Despite this negative result, 
lack of strategy-proofness does not mean that a voting rule is easy to manipulate nor that agents have the information to safely do so. The required manipulation may have to be specifically tailored to the reported preferences of all other participants. 
Given that agents typically do not have access to such information in practice, we approach the incentives to manipulate from a different perspective; one that allows agents to evaluate the potential outcomes of a manipulation according to the partial information they hold and the inferences they can make based on what they observe. 

We examine the incentives to report preferences truthfully from the lenses of regret avoidance, in the sense of \cite{fernandez2020deferred}.%
\footnote{For evidence on the effect of anticipated regret on decision making see \citet[][psychology]{zeelenberg2018anticipated}, \citet[][neuroscience]{coricelli2005regret}, \citet[][experimental economics]{fioretti2018dynamic} and references therein. Other notions that account for the partial information that agents hold are also justified in terms of regret-avoidance; e.g.   
  ex-post  \citep{bergemann2008expost} and posterior equilibrium \citep{green1987posterior}. 
 Beyond the present context, regret-based incentives have helped explain behavior in 
voter turnout \citet{ferejohn1974paradox},
 auctions \citep{ozbay2007auctions}, monopoly pricing under ignorance \citep{bergemann2008pricing,bergemann2011robust}, residency matching \citep{fernandez2020deferred} and school choice \citep{chen2024regret}. See \cite{stoye2009minimax} for a wide-ranging list of applications.  
}

We assume that an agent knows their own preference over alternatives and  the voting rule used to select an outcome. Upon observing the outcome,
the agent infers which were the possible preferences profiles reported by the other agents. An agent suffers regret if the chosen report is dominated ex-post. 
A voting rule is regret-free truth-telling if it guarantees that no agent regrets   reporting their preferences truthfully. 

First, we show that if the voting rule only considers the most preferred alternative reported by each agent (i.e. it is tops-only), then regret-free truth-telling is equivalent to strategy-proofness. 
This equivalence implies that: (i) for problems with only two alternatives,  extended majority voting rules are the only regret-free truth-telling rules; and, 
(ii) there are no non-dictatorial tops-only regret-free truth-telling voting rules, on the universal domain. 
Thus, we examine non-tops-only rules when there are more than two alternatives. 
We study whether three of the most important families of non-tops-only voting methods satisfy regret-free truth-telling:
(i) maxmin methods, (ii) scoring methods, and (iii) Condorcet consistent methods. Maxmin methods select those alternatives that  ``make the least happy agent(s) as happy as possible'' \citep{rawls1971theory}. Scoring methods assign points to each alternative according to the rank it has in agents' preferences and selects one of the alternatives with the highest score. Condorcet consistent methods select the pairwise majority winner (Condorcet winner) whenever one exists.
To resolve potential multiplicity in the scoring and maxmin methods we consider two classical tie-breakers.   One is defined by picking according to the preference of a fixed agent (neutral).  The other is defined by a fixed order of the alternatives (anonymous). 

Given $n$ agents and $m$ alternatives, 
we show that all neutral maxmin rules are regret-free truth-telling. Anonymous maxmin rules are regret-free truth-telling if and only if $n\geq m-1$ or $n$ divides $m-1$. We also obtain general positive results for the negative plurality rule, a special scoring rule in which all the  rank positions get one point except the last one which gets zero. 
The results are similar to those of maxmin rules: all neutral negative plurality rules are regret-free truth-telling,  whereas an anonymous negative plurality rule is regret-free truth-telling if and only if $n \geq m-1$. 

For general scoring rules, we provide necessary and sufficient conditions for some classes of these rules to be regret-free truth-telling.  
Among others, these include $k$-approval rules. 
Notably, Borda, plurality and Dowdall rules, as well as all efficient and anonymous rules fail to be regret-free truth-telling.

We find that Condorcet consistent rules are incompatible with regret-free truth-telling under a mild monotonicity condition. 
The monotonicity condition says that if an alternative is ranked below the outcome of the rule for an agent and he changes his preferences modifying only the ordering of alternatives ranked above the outcome, then such alternative continues not to be chosen.  
In particular, we get that  the six  famous Condorcet consistent rules associated with the names of Simpson, Copeland, Young, Dodgson, Fishburn and Black (in both anonymous and neutral versions) are monotone, and therefore not regret-free truth-telling.\footnote{For the anonymous Simpson and Copeland rules, these  results have been previously obtained by \cite{endriss2016strategic}.}  We also show that a family of non-monotone Condorcet consistent rules, the successive elimination rules,  are not regret-free truth-telling on the universal domain.

Finally, for the case with two agents and three alternatives, we present two characterization results.
The first one says that a rule is regret-free truth-telling, efficient, and anonymous if and only if it is either a successive elimination  or an anonymous maxmin rule in which the tie-breaking device is an antisymmetric and complete (not necessarily transitive) binary relation.  The second one says that a rule is regret-free truth-telling and neutral if and only if it is a dictatorship or a maxmin rule with a specific type of tie-breaking that preserves neutrality.

\subsection{Related literature}

Two main approaches have been taken to circumvent Gibbard-Satterthwaite's impossibility theorem. The first approach restricts the domain of preferences that agents can have over alternatives \citep[see][and references therein]{barbera2011strategyproof}. 
This paper contributes to the literature following the second approach, that is, to consider different notions of strategic behavior.%
\footnote{Early examples of this approach include \cite{farquharson1969theory}'s sophisticated voting, \cite{moulin1979dominance}'s dominance-solvable voting schemes, as well as \cite{barbera1982implementability}'s protective strategies and \cite{moulin1981prudence}'s prudent strategies.}

Among the papers considering weakenings of strategy-proofness, a recent literature has emerged that incorporates the (possibly) partial information held by agents.

\cite{reijngoud2012voter} and \cite{endriss2016strategic} 
 study  
when an agent has an incentive to manipulate different voting rules subject to different information functions. Importantly, there, the concept of winner information function leads to a property equivalent to regret-free truth-telling.

\cite{gori2021manipulation} introduces the notion of WMG-strategy-proofness, where an individual's information about the preferences of others is limited to, for every pair of alternatives,  the number of
people preferring the first alternative to the second one. \cite{gori2021manipulation} presents a positive result showing a class of Pareto optimal, WMG-strategy-proof
and non-dictatorial voting functions; and a negative result proving that, when at least three alternatives are considered, no
Pareto optimal and anonymous voting function is  WMG-strategy-proof. 
 
 \cite{troyan2020obvious} introduce the concept of obvious manipulation in the context of market design, while \cite{aziz2021obvious} and \cite{arribillaga2024obvious} apply it in the context of voting.%
 \footnote{
\cite{troyan2020obvious} assume that an agent knows the possible outcomes of the mechanism conditional on  his own declaration of preferences,  and define a deviation from the truth to be an obvious manipulation if either the best possible outcome under the deviation is strictly better than the best possible outcome under truth-telling, or the worst possible outcome under the deviation is strictly better than the worst possible outcome under truth-telling.} 
A mechanism that does not allow any obvious manipulation is called not-obviously-manipulable. 
\cite{aziz2021obvious} present a general sufficient condition for a voting rule to be not-obviously-manipulable. 
They show that Condorcet consistent as well as some other strict scoring rules are not-obviously-manipulable.
Furthermore, for the class of $k$-approval voting rules, they give necessary and sufficient  conditions for obvious manipulability. \cite{arribillaga2024obvious}, in turn, focus on tops-only rules. In two well-known classes of rules, the (generalized) median voter schemes and voting by committees,   they identify which subset of rules  do not have obvious manipulations when defined on the universal domain of preferences.

The rest of the paper is organized as follows. In Section \ref{Preliminaries}, we introduce the model and the property of regret-free truth-telling. We  show the  equivalence of regret-free truth-telling and strategy-proofness for tops-only rules in Section \ref{Regret-freeness}, where we also characterize extended majority voting rules as the only regret-free truth-telling rules when there are only two alternatives to choose from. In Section \ref{sec:maxmin}, we provide necessary and sufficient conditions for maxmin rules to be regret-free truth-telling. Section \ref{sec:scoring} we provide positive and negative results regarding scoring rules. In Section \ref{section Condorcet}, we present negative results for  Condorcet consistent rules. The special case with two agents and three alternatives is analyzed in Section \ref{two and three}, where two characterizations are presented.

\section{Preliminaries}\label{Preliminaries}

\subsection{Model}

A set of \textit{agents }$N=\{1,\ldots ,n\}$, with $n\geq 2$, has to choose
an alternative from a finite and given set $X,$ with $|X|=m \geq 2.$  Each agent $i\in N$ has a strict \textit{preference} $%
P_{i}$ 
over $X.$ 
We denote by $R_{i}$ the weak preference over $X$
associated to $P_{i};$ i.e., for all $x,y\in X$, $xR_{i}y$ if and
only if either $x=y$ or $xP_{i}y.$ Let $\mathcal{P}$ be the set of all
strict preferences over $X.$ 
A \emph{(preference) profile} is a $n$-tuple $%
P=(P_{1},\ldots ,P_{n})\in \mathcal{P}^n,$ an ordered list of $n$ preferences, one for each agent.
Given a profile $P$ and an agent $i,$ $P_{-i}$ denotes the subprofile 
obtained by deleting $P_{i}$ from $P.$ 
For each $P_{i}\in \mathcal{P},$ denote by $t_{k}(P_{i})$ to the alternative in the \emph{$k$-th position} (from bottom to top). Many times we write $t(P_{i})$ instead of $t_{m}(P_{i})$, and refer to it as the \emph{top} of $P_i$. We often also write $P_i$ as an ordered  list
\begin{equation*}
P_{i}:t_{m}(P_{i}),t_{m-1}(P_{i}),\ldots, t_{1}(P_{i}).
\end{equation*}%
 A \textit{(voting) rule} is a function $f:\mathcal{P}^n%
\longrightarrow X$ that selects for each preference profile $P\in \mathcal{P}^n$ an alternative $f(P)\in X$. 
We assume throughout that a voting rule is an onto function. 
Next, we define several classical properties that a rule  may
satisfy and that we use in the sequel.    
A rule $f$ is: 

\begin{itemize}
    \item 
\textit{strategy-proof} if agents can never induce a strictly preferred outcome by misrepresenting their preferences; namely, for
each $P\in \mathcal{P}^n,$ each $i\in N$ and each $P_{i}^{\prime }\in \mathcal{P}%
,$%
\begin{equation*}
f(P)R_{i}f(P_{i}^{\prime },P_{-i});
\end{equation*}
\item
\emph{efficient} if, for each $P\in \mathcal{P}^n$,  there is no $y\in X$ such that $yP_{i}f(P)$ for each $%
i\in N$;

\item
\emph{tops-only} if $P,P^{\prime }\in \mathcal{P}^n$ such that $%
t(P_{i})=t(P_{i}^{\prime })$ for each $i\in N$ imply $f(P)=f(P^{\prime })$;

\item 
\emph{dictatorial} if there exists $i\in N$ such that for each $P\in
\mathcal{P}^n$, $f(P)=t(P_{i})$. In a dictatorial rule, in each profile of preferences, the same agent selects his most preferred outcome;

\item
 \emph{unanimous} if $t(P_{i})=x$
for each $i\in N$ imply $f(P)=x$. Unanimity is a natural and weak form of efficiency: if all agents consider an
alternative as being the most-preferred one, the rule should select it;  

\item
 \emph{anonymous} if for
each $P\in \mathcal{P}^n$ and each bijection $\pi :N\longrightarrow N,$ $%
f(P)=f(P^{\pi })$ where for each $i\in N,$ $P_{i}^{\pi }=P_{\pi (i)}$. Anonymity requires that the rule treats all agents equally because the social outcome is
selected without paying attention to the identities of the agents; 

\item
\emph{neutral} if for each $P\in \mathcal{P}^n$ and each bijection $\pi
:X\longrightarrow X,$ $\pi (f(P))=f(\pi P)$ where $\pi P_{i}:\pi (t(P_{i})),\pi
(t_{m-1}(P_{i})),...,\pi (t_{1}(P_{i})).$ 

\end{itemize}

In general,
the axioms of anonymity and neutrality are incompatible for voting rules.  A classical way to address such incompatibility is to consider rules
defined in two stages as follows:\footnote{``In practice, we will be happy with (voting) correspondences that respect the three principles (efficiency, anonymity and neutrality). If a deterministic election is called for, we will use either a non-anonymous tie-breaking rule or a non-neutral one'' \citep[see][p.234]{moulin1991axioms}.} 
\begin{enumerate}

\item[(i)] First, consider a \emph{voting correspondence} $%
\mathcal{Y}:\mathcal{P}^n \longrightarrow 2^{X} \setminus  \{ \emptyset \}$ that for
each preference profile $P \in \mathcal{P}^n$ chooses a (non-empty) subset  $\mathcal
Y(P) \subseteq X,$ and assume that $%
\mathcal{Y}$ satisfies both  anonymity and neutrality.

\item[(ii)] Second, given $P\in \mathcal{P}^n$, consider a strict order  on $X$ and choose the maximal  element according to that order in $\mathcal{Y}(P)$. There are two classical selections of such an order, one to preserve
anonymity and the other to preserve neutrality:

\begin{enumerate}

\item[(a)] The strict order $\succ $ is independent of $P$ and is part of
the rule's definition. In this case anonymity is preserved and the rule is
defined by\footnote{Throughout the paper, given a strict order $>$ defined on a set $A$ and a subset $B \subseteq A,$ we denote by $\max_> B$ to the maximum element in set $B$ according to order $>$.}
\begin{equation}\label{cond:anon}
f(P)=\underset{\succ }{\max } \ \mathcal{Y}(P).
\end{equation}

\item[(b)] There exists an agent $i \in N$ such that, for each $P\in \mathcal{P}
$, the strict order we consider is $P_{i}.$ In this case, neutrality is
preserved and the rule is defined by
\begin{equation}\label{cond:neutral}
f(P)=\underset{P_{i}}{\max } \ \mathcal{Y}(P).
\end{equation}

\end{enumerate}

\end{enumerate}

From Section \ref{maxmin and scoring}  onward, we study rules that can be defined by this two-stage procedure. This way to define a rule is flexible enough to encompass many well-known and long-studied families of rules.

\subsection{Regret-free truth-telling}

We now adapt the notion of regret-free truth-telling \citep{fernandez2020deferred} to the voting context. A regret-free truth-telling rule provides incentives to report preferences truthfully if agents want to avoid regret.
When observing an outcome, an agent regrets his report if it is dominated ex-post. That is, if there is an alternative report that would have guaranteed him the same or a better outcome for any preference report of others, consistent with his observation.
We can then define a rule to be regret-free truth-telling as one in which no agent ever regrets reporting their preferences truthfully.  
Equivalently, given an observed outcome after reporting his preferences truthfully, if the agent could have done better through an alternative report for some configuration of reports of others consistent with his observation, then, that same alternative report could have done worse for another configuration of reports consistent with his observation. %
Formally,  

\begin{definition}
The rule  $f:\mathcal{P}^n%
\longrightarrow X$ is \textbf{regret-free truth-telling } if for each  $i \in N,$ each $P \in \mathcal{P}^n,$ and each $P_i' \in \mathcal{P}$ such that  $f(P_{i}^{\prime },P_{-i})P_{i}f(P),$ there is $P^\star_{-i} \in \mathcal{P}^{n-1}$ such that 
$$f(P_i, P_{-i}^\star)=f(P) \text{  and  }f(P_i,P_{-i}^\star)P_if(P_i',P_{-i}^\star).$$
\end{definition}

\section{Tops-only rules and the case of two alternatives}\label{Regret-freeness}

It is clear that if a rule is strategy-proof, then it is also regret-free truth-telling. Our first result states that the converse is also true for tops-only rules.%
\footnote{The result holds regardless of whether the voting rule is onto.}

\begin{proposition}\label{Prop Tops-only}
If a rule is regret-free truth-telling and tops-only, then it is strategy-proof. 
\end{proposition}
\begin{proof}
Let $f:\mathcal{P}^n \longrightarrow X$ be a regret-free truth-telling and tops-only rule. Assume $f$ is
not strategy-proof. Then, there are $P\in \mathcal{P}^{n},$ $i\in N,$ and $%
P_{i}^{\prime }\in \mathcal{P}$ such that $f(P_{i}^{\prime
},P_{-i})P_{i}f(P).$ Let $\widetilde{P}_{i}\in \mathcal{P}$ be such that $t(%
\widetilde{P}_{i})=t(P_{i})$ and $t_1(\widetilde{P}_{i})=f(P).$ Since $f$ is tops-only, $f(%
\widetilde{P}_{i},P_{-i})=f(P).$ Therefore, since $t_1(\widetilde{P}_{i})=f(P),$ it follows that 
\begin{equation}\label{new1}
f(P_{i}^{\prime },P_{-i})\widetilde{P}_{i}f(\widetilde{P}_{i},P_{-i}).
\end{equation}
Let $P_{-i}^{\star}\in \mathcal{P}^{n-1}$ be such that $f(\widetilde{P}_{i}, P_{-i}^{\star})=f(\widetilde{P}_{i},P_{-i}).$ Since $%
t_1(\widetilde{P}_{i})=f(P)=f(%
\widetilde{P}_{i},P_{-i})$, we
have 
\begin{equation}\label{new2}
f(P_{i}^{\prime },P_{-i}^{\star })\widetilde{R}_{i}f(\widetilde{P}%
_{i},P_{-i}^{\star }).
\end{equation}
By \eqref{new1} and \eqref{new2}, $f$ is not regret-free truth-telling.
\end{proof}

The result, in turn, leads to a complete characterization of the class of regret-free truth-telling rules for the case of two alternatives.
In order to present it, we first need to define the
family of extended majority voting rules on $\{x,y\}$.\footnote{These rules are equivalent to the ones presented in \cite{moulin1980strategy}, where fixed ballots are used to describe them instead of committees.} Fix $w\in \{x,y\}$
and let $2^{N}$ denote the family of all subsets of $N,$ referred to as
\emph{coalitions}. A family $\mathcal{C}_{w}\subseteq 2^{N}$ of coalitions is a
\textit{committee for} $w$ if it satisfies the following monotonicity
property: $S\in \mathcal{C}_{w}$ and $S\subsetneq T$ imply $T\in \mathcal{C}%
_{w}$. The elements in $\mathcal{C}_{w}$ are called \emph{winning coalitions (for $w$)}.

\begin{definition}
A rule $f:\mathcal{P}^n\longrightarrow \{x,y\}$ is an \textbf{extended majority voting} rule if there is a committee $\mathcal{C}_x$ for $x$ with the property that, for each $P \in \mathcal{P}^n,$ 
$$f(P)=x \text{ if and only if } \{i \in N \ : \ t(P_i)=x\} \in \mathcal{C}_x.$$
\end{definition}

The following corollary provides the characterization result.

\begin{corollary}\label{cor strategy-proof}
Assume $m=2.$ Then, 
\begin{enumerate}[(i)]
\item A rule is regret-free truth-telling if and only if it is
strategy-proof;
\item A rule is regret-free truth-telling if and only it is an
extended majority voting rule.
\end{enumerate}
\end{corollary}
\begin{proof}
(i) If $f$ is strategy-proof it is clear that $f$ is regret-free truth-telling. If $f$ is regret-free truth-telling, since when  $m=2$ every rule is tops-only, $f$ is strategy-proof by Proposition \ref{Prop Tops-only}.
(ii) It follows from (i) and \cite{moulin1980strategy}. 
\end{proof}

\begin{corollary}
Assume $m > 2.$ A rule is regret-free truth-telling 
and tops-only if and only if  it is a dictatorship.
\end{corollary}
\begin{proof}
It follows from Proposition \ref{Prop Tops-only} and
Gibbard-Satterthwaite's Theorem.
\end{proof}

\section{Non-tops-only rules and more than two alternatives}
\label{maxmin and scoring} 

From now on, we assume that $m>2.$ 

\subsection{Maxmin rules}\label{sec:maxmin}

Voting rules in the class of maxmin methods select those alternatives that  ``make the least happy agents as happy as possible'' \citep{rawls1971theory}.  
Given $P\in \mathcal{P}^n$ and $x \in X,$ the \emph{minimal position} of $x$
according to $P$ is defined by
\begin{equation*}
mp(x,P)=\min \{k:\text{there exists }i\in N\text{ such that }x=t_{k}(P_{i})\}.
\end{equation*}
An alternative is a \emph{maxmin winner} if there is no other alternative with higher minimal position. We denote the set of   \emph{maxmin winners} according to $P$ as $\mathcal{M}(P)$. Namely,
\begin{equation*}
\mathcal{M}(P)=\{x \in X :mp(x,P)\geq mp(y,P)\text{ for each }y\in X\}.
\end{equation*}

The idea of making  the least happy agents as happy as possible is captured by rules that pick, for each preference profile, a maxmin winner for that profile. We study both anonymous and neutral versions of these rules. Formally,

\begin{definition}
A rule  $f:\mathcal{P}^{n}\longrightarrow X$ is 
\begin{enumerate}[(i)]
\item  \textbf{$\boldsymbol{A}$-maxmin} if there is a strict order $\succ$ on $X$ such that, for each $P\in \mathcal{P}^{n},$ $$f(P)=\underset{\succ }{\max } \ \mathcal{M}(P).$$

\item \textbf{$\boldsymbol{N}$-maxmin}  if there is an agent $i \in N$  such that, for each $P\in \mathcal{P}^{n},$ $$f(P)=\underset{P_{i}}{\max } \ \mathcal{M}(P).$$
\end{enumerate}
\end{definition}

The following theorem summarizes the positive results concerning regret-free truth-telling for these rules:

\begin{theorem}\label{theo maxmin rules}
\;
\begin{enumerate}[(i)]
\item An $A$-maxmin rule is regret-free truth-telling if and only if $n\geq m-1$ or $n$ divides $m-1.$
\item Any  $N$-maxmin rule is regret-free truth-telling. 
\end{enumerate}
\end{theorem}
\begin{proof}
See  Appendix \ref{proof theo maxmin}. 
\end{proof}

The idea behind the proof can be explained as follows. First, we analyze $A$-maxmin rules. 
 To see that regret-free truth-telling implies  $n \geq m-1$ or $n$  divides $m-1$, consider an $A$-maxmin rule in which the tie-breaking has $a$ and $b$ as first and second alternatives, respectively, and alternative $z$ as the last one. 

Suppose that agent 1 has the preference $P_1: b,\dots,a,z,\dots$, that he reports truthfully and that the outcome of the rule is $a$.
Among the subprofiles consistent with $a$, there are some in which alternatives $a$ and $b$ are maxmin winners with respect to that profile, and the outcome is $a$ due to the tie-breaking. 
For any such subprofile, agent 1 could have generated a better outcome by interchanging the order of alternatives $a$ and $z$ in his report. Namely, $a$ would not longer be among the maxmin winners, whereas $b$ would remain a maxmin winner and be selected by the rule due to the tie-breaking. 
The subprofiles where both $a$ and $b$ are maxmin winners, are not the only subprofiles consistent with observing the outcome of the rule being $a$. However, for any subprofile consistent with outcome $a$, agent's 1 misrepresentation by interchanging $a$ and $z$ results in maxmin winners that are either $z$ or an alternative that is at least as good $a$ according to agent $1$'s preferences.  

 The main claim is to prove that if neither $n \geq m-1$ nor $n$ divides $m-1$,  
 for any consistent subprofile there exists a maxmin winner different from $z$ when agent 1 misrepresents his preferences by interchanging $a$ and $z$. Therefore, as $z$ is the last alternative in the tie-breaking and, as we said, the maxmin winners different from $z$ under the misrepresentation and any consistent subprofile  are always at least as good as $a$ in the true preference, the agent regrets truth-telling.

 Next, we argue that  $A$-maxmin rules such that $n\geq m-1$ or $n$  divides $m-1$  and  $N$-maxmin rules are regret-free truth-telling. First, in both cases we prove that in any profitable misrepresentation the outcome of the rule under truth-telling, $f(P)$,  has to be in a lower position than the  one it has in the true preference of the agent. Therefore, in the misrepresentation there is  an alternative  $x$ that is less preferred than $f(P)$  in the true preference that is lifted to a position  greater or equal to the one $f(P)$  has in the true preference.

For each case, we construct a subprofile of the other agents such that: (i) it is consistent  with the true preference of the agent and the outcome $f(P)$, and (ii) $x$ is the only maxmin winner under the misrepresentation and the subprofile. Thus, the agent does not regret truth-telling.
 
 For an $A$-maxmin rule such that $n\geq m-1$ the  subprofile for the other agents is such that $x$ and $f(P)$ are the first and second alternatives for each agent, respectively and for each  alternative different from $x$ and $f(P)$ there is an agent that has that alternative as his bottom alternative. 
 
 For an $A$-maxmin rule such that $n$ divides $m-1$, the definition of the subprofile is more involved but  a similar argument to the one presented in the previous case can be made.
 
 Finally, for a $N$-maxmin rule the subprofile for the other agents is such that $f(P)$ and $x$ are the first and second alternatives for each agent, respectively, while all the other alternatives keep their relative rankings. The technical details can be found in Appendix \ref{proof theo maxmin}.

\subsection{Scoring rules} \label{sec:scoring}

Next, we present the family  of scoring rules. Given  $P\in \mathcal{P}^n$ and $x \in X,$ let $N(P,k,x)=\{i\in
N:t_{k}(P_{i})=x\}$ be the set of agents that have alternative $x$ in the $k$-th position (from bottom to top) in their preferences, and let $n(P,k,x)=\left\vert
N(P,k,x)\right\vert.$ Let $s_k$ be the score associated to the  $k$-th position (from bottom to top) with  $s_{1}\leq s_{2}\leq \ldots \leq s_{m}$  and  $s_1 < s_m.$  The \emph{score of $x\in X$ according
to $P$} is defined by
\begin{equation*}
s(P,x)=\underset{k=1}{\overset{m}{\sum }}[s_{k}\cdot n(P,k,x)].
\end{equation*}
The set of \emph{scoring winners} according to $P$ is
\begin{equation*}
\mathcal{S}(P)=\{x\in X:s(P,x)\geq s(P,y)\text{ for all }y\in X\}.
\end{equation*}

\begin{definition} A rule  $f:\mathcal{P}^{n}\longrightarrow X$ is
\begin{enumerate}[(i)]
\item \textbf{$\boldsymbol{A}$-scoring}  associated to $s_{1}\leq s_{2}\leq \ldots \leq s_{m}$ if there is an order  $\succ$ on $X$ such that, for each  $P\in \mathcal{P}^{n},$ 
$$f(P)=\underset{\succ }{\max } \ \mathcal{S}(P).$$

\item  \textbf{$\boldsymbol{N}$-scoring}   associated to $s_{1}\leq s_{2}\leq \ldots \leq s_{m}$ if there is an agent $i \in N$ such that, for each  $P\in \mathcal{P}^{n},$ 
$$f(P)=\underset{P_{i}}{\max } \ \mathcal{S}(P).$$
\end{enumerate}
\end{definition}

\begin{remark}
If an $A$-scoring rule is efficient, then $s_{m-1}<s_{m}$.
\end{remark}

\begin{remark} Some of the most well known scoring rules are:

\begin{enumerate}[(i)]

\item the \textbf{Borda} rule, in which $s_{k}=k$ for $k=1,\ldots ,m$;

\item the \textbf{Dowdall} rule,  
in which $s_{k}=\frac{1}{m-k+1}$ for $k=1,\ldots ,m$;

\item the \textbf{$\boldsymbol{k}$-approval} rules, in which $0=s_{1}=s_{2}=\ldots=s_{m-k},$ $s_{m-k+1}=\ldots=s_{m-1}=s_{m}=1$ for some $k$ such that $m-1 \geq k\geq 1$, i.e., the top $k$ scores are 1 and the rest are 0. In these rules, agents are asked to name their $k$ best alternatives, and the alternative with most votes wins. 

Within these rules, two subclasses stand out:

\begin{enumerate}

\item[(iii.a)] the \textbf{plurality} rule, where $k=1$. Therefore $s_{1}=s_{2}=\ldots=s_{m-1}=0$ and $s_{m}=1$;

\item[(iii.b)] the \textbf{negative plurality} rule, where $k=m-1$. Therefore $s_{1}=0$ and $s_{2}=\ldots=s_{m-1}=s_{m}=1$. 

\end{enumerate}

\end{enumerate}

\end{remark}

The next theorem considers the cases in which $s_{m-1}<s_m$  and allows us to present conclusive results about efficient $A$-scoring rules, the Borda rule, the Dowdall rule, and the plurality rule.

\begin{theorem}\label{k star m-1}

 Assume $n>2.$ Then,   no (anonymous or neutral) scoring rule  with $s_{m-1}<s_m$ is regret-free truth-telling.

\end{theorem}
\begin{proof} See Appendix \ref{proof of theo scoring k star m-1}. 
\end{proof}

We sketch the proof of Theorem \ref{k star m-1} for the case of an odd number of voters. Consider a scoring rule in which the tie-breaking has $a$ and $b$ as first and second alternatives, respectively, in the anonymous case or agent 1 in the neutral case. Let $P \in \mathcal{P}^n$ be given by the following table:
\begin{center}
$%
\begin{array}{ccc:ccc:ccc}
P_{1} & P_{2} & P_{3} & P_4 & \cdots & P_{t+3} & P_{t+4} & \cdots & P_{2t+3} \\
\hline

a & c & b & a & \cdots & a & b & \cdots & b \\
c & b & a & b & \cdots & b & a & \cdots & a \\
b & a & c & c & \cdots & c & c & \cdots & c \\
\vdots & \vdots & \vdots & \vdots & \cdots & \vdots & \vdots & \cdots & \vdots\\
\multicolumn{3}{c}{} & \multicolumn{3}{c}{$\upbracefill$} &\multicolumn{3}{c}{$\upbracefill$}\\
\multicolumn{3}{c}{} & \multicolumn{3}{c}{t \text{ agents}} &\multicolumn{3}{c}{t \text{ agents}}\\
\end{array}
$
\end{center}
Then, the alternatives $a$ and $b$ are scoring winners and the tie-breaking chooses $a$. Now, as $s_{m-1} < s_m$, agent 2 can manipulate the rule by interchanging $b$ and $c$ in his report, because the outcome of the rule then changes to $b$. Furthermore, for any profile where agent 2 declares his true preference and the outcome of the rule is $a$, if agent 2 interchanges $b$ and $c$ in his report then the outcome of the rule is at least as good as $a$. Hence, agent 2 regrets truth-telling. 

\begin{corollary}\label{cor Borda et al}
Assume $n>2.$  Then, both anonymous and neutral versions of Borda, plurality, and Dowdall rules are not regret-free truth-telling. Moreover, no efficient $A$-scoring rule is regret-free truth-telling.
\end{corollary}

Next, we focus $k$-approval rules. The plurality rule, where $k=1$, is covered under Theorem \ref{k star m-1}. Theorem \ref{scoring k=1} presents the results for the negative plurality rules where $k=m-1$. Lastly, theorem \ref{theo approval} tackles the case of $1<k<m-1$.

\begin{theorem}\label{scoring k=1}
\;
\begin{enumerate}[(i)]

\item An $A$-negative plurality rule is regret-free truth-telling if and only if $n \geq m-1$.\footnote{For $A$-negative plurality rules, Theorem 6 in \cite{reijngoud2012voter} presents a sufficient (though not necessary) condition ($n+2 \geq 2m$) guaranteeing regret-free truth-telling. We present an independent proof that encompasses their result as well.}

\item Any $N$-negative plurality rule is regret-free truth-telling.

\end{enumerate}

\end{theorem}
\begin{proof} See  Appendix \ref{proof of theo scoring k=1}. 
\end{proof}

The proof of Theorem \ref{scoring k=1} follows similar ideas to the ones that describe the proof  of Theorem \ref{theo maxmin rules}. 

\begin{theorem}\label{theo approval}
Assume $n>2$ and $1<k<m-1$. Then, 
\begin{enumerate}[(i)]

\item An anonymous $k$-approval rule  is  regret-free truth-telling if and only if $kn=m(n-1)+1.$%
\footnote{
Theorem \ref{theo approval} (i) contradicts Theorem 2 in \cite{endriss2016strategic} which states that there is no anonymous approval rule that satisfies regret-free truth-telling. The proof in \cite{endriss2016strategic} assumes that there is a preference profile  
in which the outcome chosen by the rule does not receive maximal score by some agent. 
However, this assumption cannot be met in the case where  $kn>m(n-1).$}

\item A neutral $k$-approval rule is  regret-free truth-telling if and only if $kn>m(n-1).$
\end{enumerate}
\end{theorem}

\begin{proof} See  Appendix \ref{proof of theo approval}.
\end{proof}

We have presented Theorem \ref{theo approval} here to provide a comprehensive and transparent description of the results concerning $k$-approval voting. As we will see in Appendix \ref{proof of theo approval}, the proof of this theorem follows from Theorems \ref{scoring positive} and \ref{scoring negative}, which will be introduced later. The main ideas of these proofs, which also clarify the proof of Theorem \ref{theo approval}, will be deferred until their statements are presented.

To study the remaining scoring rules, we need to identify the highest position where the score is not maximal. Given scores  $s_{1}\leq s_{2}\leq \ldots \leq s_{m},$ let  $k^{\star}$ be such that  $s_{1}\leq s_{2}\leq s_{k^{\star }}<s_{k^{\star }+1}=\ldots=s_{m}.$ Observe that, by definition, $k^\star \in \{1, 2, \ldots, m-1\}.$

From now on, we assume that $k^\star$ is such that $1<k^\star <m-1.$\footnote{Case $k^\star=1$ corresponds to negative plurality (that can be defined by simply asking $s_1 <s_2=\ldots=s_m$, without the need of $s_1=0$ and $s_m=1$) and is analyzed in Theorem \ref{scoring k=1}. The case $k^\star=m-1$ is analyzed in Theorem \ref{k star m-1}.} Next, we present some results for scoring rules by means of two  complementary theorems, one  of which can be considered as positive and the other one as negative. 
Theorem \ref{scoring positive}  focuses on the case  $k^{\star }n<m,$ which encompasses the class of scoring rules where, in any preference profile, there is always an alternative that gets maximal score.  This positive result gives a necessary and sufficient condition for an $A$-scoring rule to be regret-free truth-telling and also states that any $N$-scoring rule is regret-free truth-telling.

\begin{theorem}\label{scoring positive}
Assume that $n>2$ and $k^\star n<m.$ Then,
\begin{enumerate}[(i)]

\item An $A$-scoring rule is  regret-free truth-telling if and only if $k^\star n=m-1.$

\item  Any $N$-scoring rule is regret-free truth-telling.
\end{enumerate}
\end{theorem}
\begin{proof}
See  Appendix \ref{proof theo scoring positive}. 
\end{proof}

The main idea of the proof of Theorem \ref{scoring positive} is the following. Assume that $n>2$ and $k^\star n<m.$ First, we analyze $A$-scoring rules. To see that regret-free truth-telling implies $k^\star n=m-1$, consider a scoring rule in which the tie-breaking has $a$ and $b$ as first and second alternatives, respectively, and $z$ as the last one.  Assume $k^\star n<m-1$. Then, for any profile of preferences there are at least two alternatives above position $k^\star$ in the preference of each agent.  So, there are at least two score winners (with score equal to $n \cdot s_m$). Let $P \in \mathcal{P}^n$ be such that  $P_1: b, \ldots, a,z,\ldots$ where $z$ is in the $k^\star$-th position and $P_j:a,b,\ldots$ for each $j \in N \setminus \{1\}$. Now, the argument  to show that agent 1 regrets truth-telling is similar to the one used for Theorem \ref{theo maxmin rules} part (i).

Next, we argue that  $A$-scoring rules such that $k^\star n=m-1$ and $N$-scoring rules are regret-free truth-telling. 
In each case we prove that in any  
profitable misrepresentation of preferences the outcome of the rule under truth-telling, $f(P)$ (that must be above the $k^\star$-th position in the true preference), is moved to the $k^\star$-th position or below. Therefore, in the misrepresentation there is  an alternative  $x$, that is in the $k^\star$-th position  or below in the true preference of the agent, that is lifted to a position  greater than $k^\star$.

 For an $A$-scoring rule such that $k^\star n =m-1$, consider a  profitable manipulation for an agent and a subprofile for the other agents such that $x$ and $f(P)$ are the first and second alternatives for each agent, respectively, and for each  alternative different from $f(P)$ and above  the $k^\star$-th position in the true preference of the agent, there is an agent that has that alternative in the $k^\star$-th position or below. By the definition of the subprofile, $f(P)$ is the only alternative whose score is $n \cdot s_m$ and, therefore, it is the outcome under the true preference and that subprofile. Furthermore, as $f(P)$ is in position $k^\star$ or below in the misrepresentation, by the definition of the subprofile the outcome under the misrepresentation and the subprofile must necessarily  be an alternative  that is in the $k^\star$-th position or below in the true preference. Thus, the agent does not regret truth-telling.   
 
 For a $N$-scoring rule, first we prove that the tie-breaking agent does not manipulate. Then, consider a  profitable manipulation for an agent and a subprofile for the other agents  such that $x$ and $f(P)$ are the first and second alternatives for each agent, respectively. By the definition of the subprofile and the tie-breaking agent, $f(P)$ is the outcome under the true preference and that subprofile (because in such profile $f(P)$ has maximum score $n \cdot s_m$ and $x$ does not and the tie-breaking agent prefers $f(P)$ to any alternative different from $x$). Furthermore, the outcome under the misrepresentation and the subprofile is $x$ since it obtains score $n \cdot s_m$ and it is the first alternative for the tie-breaking agent.

  Theorem \ref{scoring negative} below  gives a negative result for the case  $k^{\star}n\geq m$ when  $s_{k^\star-1}=s_{k^\star}$.  When $s_{k^\star-1} \neq s_{k^\star},$ we believe that the existence of regret-free truth-telling rules in the class  depends sensibly on the specific scores defining each rule.

\begin{theorem}\label{scoring negative} 
Assume that $n>2$ and $k^\star n \geq m.$ Then, there is no regret-free truth-telling scoring rule (neither anonymous nor neutral) with  $s_{k^\star-1}=s_{k^\star}.$ 
\end{theorem}
\begin{proof}
See  Appendix \ref{proof theo scoring negative}. 
\end{proof}

The main idea of the proof of Theorem \ref{scoring negative} is the following. Assume that $n>2$,  $k^\star n \geq m$, and $s_{k^\star-1}=s_{k^\star}$.   Consider a scoring rule in which the tie-breaking has $a$ and $b$ as first and second alternatives, respectively, in the case of an anonymous rule, or in which agent 1 breaks ties when the rule is neutral. As $\mathit{k^{\star }}n\geq m$, it follows that  $k^{\star }(n-1)\geq m-k^{\star }$, so we can consider profile $P \in \mathcal{P}^n$  such that  $P_2:  \ldots, b, a,\ldots$ where $b$ is in the $k^\star$-th position and $P_j:a,b,\ldots$ for each $j \in N \setminus \{2\}$ and for each alternative above  $b$ in the preference of agent  2 there is another agent that has that alternative below the $k^{\star }$-th position. In this profile, as $s_{k^\star-1}=s_{k^\star}$, alternative $a$ has a score greater than or equal to the score of any alternative above $a$ for agent 2. Then, by the tie-breaking, the rule selects something worse than $b$ for agent 2. Then, agent 2 can manipulate the rule by misrepresenting his preferences by interchanging the alternative in the $(k^\star +1)$-th position and  $b$, because now $b$ is a scoring winner, $a$ is not, and thus the outcome is $b$. Furthermore, for any profile where agent 2 declares his true preference and the outcome of the rule is $f(P)$ (that is below $b$ for agent 2), if agent 2 misrepresents interchanging as said previously, then the outcome of the rule is at least as good as $f(P)$. Hence, agent 2 regrets truth-telling.

The previous theorem extends the result of Theorem 3 in  \cite{reijngoud2012voter} to the case $n=3$ and also to the neutral scoring rules (their result only applies when $n>3$ in the anonymous case). Our proof is independent of theirs.

\subsection{Condorcet consistent rules}
\label{section Condorcet}

Let $P\in \mathcal{P}^n$ and consider two alternatives $a,b\in X.$  Denote by $C_{P}(a,b)$ the number of agents that prefer $a$ to $b$ according to $P,$ i.e., $C_{P}(a,b)=|\{i \in N : aP_{i}b\}|.$ An alternative $a \in X$ is a \emph{Condorcet winner} according to $P$ if for each alternative $b\in X \setminus \{a\},$
\begin{equation}\label{condorcetwinner}
C_{P}(a,b)>C_{P}(b,a).
\end{equation}
Notice that a  Condorcet winner may not always exist but when it does, it is unique. If \eqref{condorcetwinner} holds with weak inequality for each alternative $b\in X \setminus \{a\},$ then $a$ is called a \emph{weak Condorcet winner.}  

\begin{definition}
A rule $f:\mathcal{P}^{n}\longrightarrow X$ is \textbf{Condorcet consistent} if it chooses the Condorcet winner whenever it exists.  
\end{definition}

Next, we introduce a mild monotonicity condition 
which says that if an alternative is below the outcome for an agent and he changes his preferences modifying only the ordering of alternatives above the outcome, then such alternative continues not to be chosen. Formally, 
\begin{definition}
Let $P_{i}, P_i' \in \mathcal{P}$ and let $a \in X$ be such that $a=t_{k}(P_{i}).$ We say that $P_{i}^{\prime }$ is a \textbf{monotonic transformation of $\boldsymbol{P_{i}$ with respect to $a}$} if $t_{k^{\prime
}}(P_{i})=t_{k^{\prime }}(P_{i}^{\prime })$ for each $k^{\prime }\leq k.$ A rule $f:\mathcal{P}^n \longrightarrow X$ is \textbf{monotone} if, for each $P\in
\mathcal{P}^{n},$ each $i \in N,$ and each $b \in X$ such that $f(P)P_{i}b,$
\begin{equation*}
f(P_{i}^{\prime },P_{-i})\neq b
\end{equation*}%
for each $P_{i}^{\prime }$ that is a monotonic transformation of $P_{i}$ with
respect to $f(P)$.
\end{definition}
Notice that when $P'_i$ is a monotonic transformation of $P_i$ with respect to $a$, then $C_P(x,a)=C_{(P_i',P_{-i})}(x,a)$ for each $x \in X \setminus \{a\}.$ Thus, our monotonicity condition is fully compatible with Condorcet consistency.

Furthermore, our notion of monotonicity is weaker than the well-known Maskin monotonicity. Remember that $P_{i}^{\prime } \in \mathcal{P}$ is a \emph{Maskin monotonic transformation} of $P_{i} \in \mathcal{P}$ with
respect to $a \in X$ if $xP_{i}^{\prime }a$ implies  $xP_{i}a.$ Then,
$f:\mathcal{P}^n \longrightarrow X$ is \emph{Maskin monotonic} if, for each $P\in
\mathcal{P}^{n}$,
$f(P_{i}^{\prime },P_{-i})=f(P)$
for each $P_{i}^{\prime } \in \mathcal{P}$ that is a Maskin monotonic transformation of $P_{i}$ with
respect to $f(P).$ %
It is clear that a monotonic transformation of $P_{i}$ (according
to our definition) is a Maskin monotonic transformation of $%
P_{i}.$ 

Besides the intrinsic appeal of our monotonicity condition, this weakening of Maskin's property is necessary since Maskin's monotonicity is incompatible with Condorcet consistency. 
Our mild monotonicity requirement on Condorcet consistent rules leads to the following negative result concerning regret-free truth-telling.

\begin{theorem}\label{theo Condorcet consistent}
Assume $n\notin\{2,4\},$ or $n=4$ and $m>3.$ Then, there is no Condorcet consistent, monotone
and regret-free truth-telling rule.
\end{theorem}
\begin{proof}
See  Appendix \ref{proof theo Condorcet consistent}. 
\end{proof}

 The main ideas behind the  proof of Theorem \ref{theo Condorcet consistent}  can be illustrated in the three-agent case, sketched next. Consider profile $P \in \mathcal{P}^3$ given by the following table:
\begin{center}
$
\begin{array}{cccc}
P_1 & P_2 & P_3 \\
\hline
a & b & c  \\
b & c & a  \\
c& a & b  \\
\vdots & \vdots & \vdots \\
\end{array}
$
\end{center}
In this profile there is no Condorcet winner. Let $f(P)$ be the chosen alternative under profile $P$. We can assume, w.l.o.g., that $f(P)$ is worse than $b$ for agent 1. Now, agent $1$ can manipulate the rule by interchanging $a$ and $b$, because the outcome of the rule then changes to $b$. Furthermore, for any profile where agent 1 declares his true preference and the outcome of the rule is $f(P)$ (that is bellow $b$ for agent 1), if agent 1 misrepresents his preferences by interchanging $a$ and $b$, then by monotonicity the outcome of the rule is at least as good as $f(P)$. Hence, agent 1 regrets truth-telling.

\begin{remark}
When $n=4$ and $m=3,$ the previous impossibility result does not apply. Let $X=\{a,b,c\}$. Given $P\in \mathcal{P}^n$ and $x \in X$, let $bottom(P,x)$ be the number of agents that have $x$ in the bottom of their preferences. Now, consider a rule $f: \mathcal{P}^4 \longrightarrow X$ that, for each $P\in \mathcal{P}^4$, selects the Condorcet winner when it exists and, otherwise, the tie-breaking   $a \succ b \succ c$ is used to choose an alternative among those  that minimize $bottom(P,\cdot)$ and are preferred by at least two agents to any other alternative that minimizes $bottom(P,\cdot)$. 
This rule is monotone  since, given $P \in \mathcal{P}^4$ and $i \in N,$ when there are three alternatives a monotonic transformation of $P_i$ with respect to $f(P)$  is different from $P_i$ only when $t_1(P_i)=f(P)$, and in this case there is no $x \in X$ such that $f(P)P_ix.$ Then, monotonicity is trivially satisfied.  To see that this rule is also regret-free truth-telling, consider profile $P \in \mathcal{P}^4$  such that, w.l.o.g., $P_1:x,y,z.$ If $f(P)=x,$ agent $1$ does not manipulate $f.$ If $f(P)=z,$ then $f(P_1',P_{-1})=t_1(P_1)$ for each $P_1' \in \mathcal{P}$ by definition of the rule, so agent $1$ cannot manipulate either. If $f(P)=y$ and agent $1$ manipulates $f$ via $P_1',$ then $f(P_1',P_{-1})=x$ and, by definition of the rule, $t_1(P_1')=y.$ Consider $P_{-1}^\star \in \mathcal{P}^{3}$ such that $P_{2}^\star: y,z,x,$ $P_3^\star=P_2^\star,$ and $P_4^\star:z,y,x.$ Then, $f(P_1, P_{-1}^\star)=y$ and $f(P_1', P_{-1}^\star)=z.$ Therefore, agent 1 does not regret truth-telling. 
\end{remark}

\medskip

Six of the most important Condorcet consistent rules are Simpson, Copeland, Young, Dodgson, Fishburn and Black rules \citep[see][]{fishburn1977condorcet}. Each one of these rules uses pairwise comparison of alternatives in a specific way in order to get a \emph{winner} alternative for each profile of preferences. Their definitions are as follows. Given  $P \in \mathcal{P}^n,$

\begin{enumerate}[(i)]

\item the \emph{Simpson score} of  alternative $a \in X$ is the minimum number $C_P(a,b)$ for $b \neq a,$
\begin{equation*}
Simpson(P,a)=\min_{b\neq a}C_P(a,b)
\end{equation*} and a \emph{Simpson winner} is an alternative with highest such score. 

\item the \emph{Copeland score} of  alternative $a \in X$ is the  number of pairwise victories minus the number of pairwise defeats against all other alternatives
\begin{equation*}
Copeland(P,a)=|\{b : C_P(a,b)>C_P(b,a)\}|-|\{b : C_P(b,a)>C_P(a,b)\}|
\end{equation*}
 and a \emph{Copeland winner} is an alternative with highest such score.
 
 \item the \emph{Young score} of  alternative $a \in X$ is the largest  cardinality of a subset of voters for which alternative $a$ is a weak Condorcet winner
\begin{equation*}
Young(P,a)=\max_{N' \subseteq N} \left\{ |N'| : \{|i \in N' : a P_i b\}| \geq \frac{|N'|}{2} \text{ for all }b \in X\setminus \{a\}\right\}
\end{equation*}
and a \emph{Young winner} is an alternative with highest such score.

\item the \emph{Dodgson score} of  alternative $a \in X,$ $Dodgson(P,a),$ is the fewest inversions\footnote{Let $P_i,P_i' \in \mathcal{P}$ and let $x,y \in X.$ $P_i'$ is an \emph{inversion of $P_i$ with respect $x$ and $y$} if  $xP_iy$ implies  $yP_i'x.$} in the preferences  in $P$ that will make $a$ tie or beat every other alternative in $X$ on the basis of simple majority, and a \emph{Dodgson winner} is an alternative with lowest such score.

\item the \emph{Fishburn partial order} on $X,$ $F_P,$  is defined as follows: $a F_P b$ if and only if for each $x \in X,$ $C_P(x,a) > C_P(a,x)$ implies $C_P(x,b) > C_P(b,x)$ and there is $w \in X$ such that $C_P(w,b)>C_P(b,w)$ and $C_P(a,w)\geq C_P(w,a).$ A \emph{Fishburn winner} is a maximal alternative for $F_P.$

\item a  \emph{Black winner} is a Condorcet winner whenever it exists and, otherwise, a \emph{Borda winner}.\footnote{A \emph{Borda winner} is an alternative with highest Borda score.}

\end{enumerate}

 An \emph{anonymous (neutral) Simpson, (Copeland, Young, Dodgson,  Fishburn, Black) rule}    always chooses a Simpson, (Copeland, Young, Dodgson, Fishburn, Black) winner and uses a fixed order (agent)   as tie-breaker when there are more than one. The following result shows that the six rules are monotonic.

\begin{corollary}\label{corollary Simpson et al}
Assume  $n>2$. Then, the Simpson, Copeland, Young, Dodgson, Fishburn and Black rules are not regret-free truth-telling, regardless of whether we consider their anonymous or neutral versions. 
\end{corollary}
\begin{proof}
See Appendix \ref{proof of corollary Simpson et al}. 
\end{proof}

Another interesting class of Condorcet consistent rules which are widely used in practice, for instance, by the United States Congress to vote upon a motion and its proposed amendments, is the class of successive elimination rules \citep[see Chapter 9 of][for more detail]{moulin1991axioms}. These rules, which consider an order among alternatives  and consist of sequential majority comparisons, are defined as follows.

\begin{definition} 
A rule $f:\mathcal{P}^{n}\longrightarrow X$ is a \textbf{successive elimination} rule with respect to
an order $\succ$ such that  $a_{1}\succ a_2\succ \ldots \succ  a_{m}$ if it operates in the following way. First, a majority vote decides to  eliminate $a_{1}$ or $a_{2}$, then a majority vote decides to eliminate the survivor from the first
round or $a_{3},$ and so on. The same order $\succ$ is used as tie-breaker in each pairwise comparison,  if necessary.
\end{definition}

It is clear that a successive elimination rule is Condorcet consistent but it may be not monotone, as the next example shows. 

\begin{example}(The successive elimination rule with respect to order $a\succ b \succ c \succ d$ is not monotone). Let $P \in \mathcal{P}^5$ be given by the following table:

\begin{center}
$%
\begin{array}{ccccc}
P_{1} & P_{2} & P_{3} & P_{4} & P_{5} \\ \hline
a & a & c & c & d \\
b & c & d & b & b \\
d & d & a & d & a \\
c & b & b & a & c
\end{array}
$
\end{center}
\vspace{5 pt}
Then, $f(P)=d.$ Now, let $P_1' \in \mathcal{P}$ be such that $P_{1}^{\prime }:  b,a, d,c.$ Then $P_1'$ is a monotonic transformation of $P_1$ with respect to $d$ but  $%
f(P_{1}^{\prime },P_{-1})=c$, so $f$ is not monotone. 
\end{example}

\begin{theorem}\label{theo succesive}
Assume $n >2.$ Then, no successive elimination rule is regret-free truth-telling.
\end{theorem}
\begin{proof}
See  Appendix \ref{proof theo succesive}. 
\end{proof}

 The main ideas behind the  proof of Theorem \ref{theo succesive}  can be found in the three-agent case, so we sketch it here. Let $f$ be a successive elimination rule with associated order $a \succ b \succ c \succ \ldots$ Consider profile $P \in \mathcal{P}^3$ given by the following table:
\begin{center}
$
\begin{array}{cccc}
P_1 & P_2 & P_3 \\
\hline
a & b & c  \\
b & c & a  \\
\vdots & a & b  \\
c & \vdots & \vdots \\
\end{array}
$
\end{center}
In this profile, the outcome of the rule is $c$.  Now, agent $1$ can manipulate the rule by interchanging $a$ and $b$, because the outcome of the rule then changes to $b$. Furthermore, as $c$ is the worst alternative in the true preference of agent 1, agent 1 regrets truth-telling.

\section{Two agents and three alternatives: characterizations}
\label{two and three}

In what follows, we focus in the  case  where we have only two agents, $%
N=\{1,2\}$, and three alternatives, $X=\{a,b,c\}.$ In this case  we can obtain characterizations of the classes of all:  (i) regret-free truth-telling and neutral,  and (ii) regret-free truth-telling, efficient, and anonymous rules. Notice that for the first characterization efficiency is not needed since it is implied by neutrality and regret-free truth-telling, as we prove next in Theorem \ref{theo neutral 2x3}.

First, observe that with two agents and three alternatives a $N$-maxmin rule coincides with:
\begin{enumerate}[(i)]

\item the $N$- negative plurality rule; and,

\item the $N$-scoring rule corresponding to $\overline{s%
}=(\overline{s}_{1},\overline{s}_{2},\overline{s}_{3})=(1,3,4).$

\end{enumerate}
The following theorem shows that $N$-maxmin and dictatorships are the only regret-free truth-telling and neutral rules. 

\begin{theorem}\label{theo neutral 2x3}
Assume $n=2$ and $m=3$. Then, a rule  is regret-free truth-telling and 
neutral if and only if it is a $N$-maxmin rule or a dictatorship.
\end{theorem}
\begin{proof}
See  Appendix \ref{proof theo neutral 2x3}. 
\end{proof}

A similar result  to the previous theorem can be obtained changing neutrality  for anonymity. As efficiency is not a consequence of  anonymity and regret-free truth-telling we require it in the next theorem.\footnote{For example, a constant rule is regret-free truth-telling and anonymous but not efficient.} 

In this case, we need to enlarge the class of $A$-maxmin rules by dropping the requirement of transitivity for the tie-breaking associated to the rules and to add the successive  elimination rules into the picture, as we did with dictatorial rules in Theorem \ref{theo neutral 2x3}. 

\begin{definition}\label{def star}
A rule $f:\mathcal{P}^{2}\longrightarrow X$ is an \textbf{$\boldsymbol{A$-maxmin$^\star}$} rule if there is an antisymmetric and complete (not necessarily transitive) binary relation $\succ^\star$ on $X$ such that, for each $P \in \mathcal{P}^2$,
$$f(P)=\underset{\succ ^{\star }}{\max } \ \mathcal{M}(P).$$  
\end{definition}
Observe that, since $n=2$, 
$\left\vert \mathcal{M}(P)\right\vert \leq 2$ and therefore $\underset{\succ ^{\star }}{\max }%
\ \mathcal{M}(P)$ is well defined. In a similar way to Definition \ref{def star} we can define the $A$-scoring$^{\star}$ rule associated to $\succ ^{\star}.$  Notice  that the $A$-maxmin$^\star$ rule associated to $\succ ^{\star}$ coincides
with the $A$-scoring$^{\star}$ rule with $\overline{s}=(\overline{s%
}_{1},\overline{s}_{2},\overline{s}_{3})=(1,3,4)$  associated to $\succ
^{\star}$.

\begin{theorem}\label{theo eff anon 2x3}
Assume $n=2$ and $m=3$. Then, a rule  is regret-free truth-telling, efficient, and anonymous  if and only if it is a successive elimination rule or an  $A$-maxmin$^\star$ rule.
\end{theorem}
\begin{proof}
See  Appendix \ref{proof theo eff anon 2x3}. 
\end{proof}

Concerning the independence of axioms in the characterizations, it is clear that  neutrality and regret-free truth-telling in Theorem \ref{theo neutral 2x3} are independent. Successive elimination rules are regret-free truth-telling but not neutral, and the rule that always chooses the bottom of agent 1 is neutral and not regret-free truth-telling.   On the other hand, in Theorem \ref{theo eff anon 2x3}, a constant rule is regret-free truth-telling, anonymous, and not efficient and a dictatorship is regret-free truth-telling, efficient, and not anonymous. Now, given order $a \succ b \succ c,$ consider the rule $f(P)= \max_{\succ} \{t(P_1),t(P_2)\}$. This rule is anonymous, efficient, and not regret-free truth-telling.

\section{Summary}

Table  \ref{tabla caracterizaciones} summarizes the characterization results when there are only two alternatives, or two agents and three alternatives. Table \ref{tabla familias} summarizes our main findings  about tops-only, maxmin, scoring, negative plurality, $k$-approval, and Condorcet consistent rules. For simplicity, we present all the results for $n>2$, although some of them also apply when $n=2$.

\begin{table}[ht] 
\small
\centering 
\begin{threeparttable}
\begin{tabular}{|c |c |c|}
\hline
$m=2$ &  \hspace{42 pt} regret-free  \ \ $\Longleftrightarrow$ \ \  ext. majority voting & Cor. \ref{cor strategy-proof} \\
 \hline \hline
$n=2, m=3$ & \hspace{22 pt} regret-free $+$ neutral \ \ $\Longleftrightarrow$ \ \ $N$-maxmin or dictatorship & Th. \ref{theo neutral 2x3} \\
\hline \hline
$n=2, m=3$ & regret-free $+$ eff. $+$ anon. \ \ $\Longleftrightarrow$ \ \ $A$-maxmin$^\star$ or succ. elim. & Th. \ref{theo eff anon 2x3} \\
\hline
\end{tabular}
\begin{tablenotes}\footnotesize
\item We use ``regret-free" to mean ``regret-free truth-telling" due to space considerations.
\end{tablenotes}
\end{threeparttable}
\caption{\emph{Characterization results with $m=2$ or $n=2$ and $m=3.$}} \label{tabla caracterizaciones}
\end{table}

\begin{table}[ht] 

\small
\centering 
\begin{threeparttable}
\begin{tabular}{| c | c | c  r l | c |}
\hline
Tops-only & \multicolumn{4}{c}{\qquad \qquad strategy-proof \ \  $\Longleftrightarrow$  \ \ regret-free} & \multicolumn{1}{|c|}{Pr. \ref{Prop Tops-only} }\\ 
\hline \hline
$A$-maxmin & \multicolumn{4}{c}{ $ \ n\geq m-1 \text{ or }n \text{ divides }m-1 \ \ \ \Longleftrightarrow$ \ \  regret-free \  \ \  \qquad \qquad} & \multicolumn{1}{|c|}{Th. \ref{theo maxmin rules}}\\ 
\hline \hline
$N$-maxmin & \multicolumn{4}{c|}{ \ \ \ \ \ \ \  \ \qquad \qquad all regret-free} & Th. \ref{theo maxmin rules} \\
\hline \hline
Scoring with & \multicolumn{4}{c|}{\multirow{2}{*}{\ \ \ \ \ \ \  \ \qquad \qquad none regret-free}} & \multirow{2}{*}{Th. \ref{k star m-1}} \\
$s_{m-1}<s_m$ \tnote{$\dag$} & \multicolumn{4}{c|}{ } &  \\
\hline \hline
$A$-negative & \multicolumn{4}{c|}{ \multirow{2}{*}{ \qquad \qquad \ \  \ \ \ \ \ \ $ \ n\geq m-1  \ \ \ \Longleftrightarrow$ \ \  regret-free \  \ \  }} & \multirow{2}{*}{Th. \ref{scoring k=1} (i)} \\
plurality & \multicolumn{4}{c|}{ } &  \\
\hline \hline
$N$-negative & \multicolumn{4}{c|}{ \multirow{2}{*}{ \ \ \ \ \ \ \  \ \qquad \qquad all regret-free}} & \multirow{2}{*}{Th. \ref{scoring k=1} (ii)} \\
plurality & \multicolumn{4}{c|}{ } &  \\
\hline \hline
Anonymous  & \multicolumn{4}{c|}{\multirow{2}{*}{ \qquad \qquad \ \  \ \ \  $ \ k=\frac{m(n-1)+1}{n}   \ \ \Longleftrightarrow$ \ \  regret-free \  \ \  }} &  \multirow{2}{*}{Th. \ref{theo approval} (i)} \\
 $k$-approval \tnote{$\ddag$} & \multicolumn{4}{c|}{} &  \\
\hline \hline
Neutral  & \multicolumn{4}{c|}{\multirow{2}{*}{ \qquad \qquad \ \  \ \ \ \ \  $ \ k>\frac{m(n-1)}{n}  \ \ \ \Longleftrightarrow$ \ \  regret-free \  \ \  }} &  \multirow{2}{*}{Th. \ref{theo approval} (ii)} \\
 $k$-approval \tnote{$\ddag$} & \multicolumn{4}{c|}{} &  \\
\hline \hline

\multirow{3}{*}{\begin{tabular}{c}
     Condorcet\\
      consistent
\end{tabular}} & Monotone & \multicolumn{3}{c|}{   $n\neq 4$ or  $m>3 \ \ \Longrightarrow$ \ \  none regret-free} & Th. \ref{theo Condorcet consistent}\\ \cline{2-6}
 & Successive  & &    \multirow{2}{*}{ \qquad \quad \   \ \ \  none regret-free} &  & \multirow{2}{*}{Th. \ref{theo succesive}}\\

 & elimination &&&&\\
\hline 
\end{tabular}
\begin{tablenotes}\footnotesize
\item We use ``regret-free'' to mean ``regret-free truth-telling'' due to space considerations. The results for tops-only, maxmin, and negative plurality are also valid when $n=2$.
\item[$\dag$]This class includes Borda, plurality, Dowdall, and efficient $A$-scoring rules (see Corollary \ref{cor Borda et al}). 
\item[$\ddag$] We exclude plurality and negative plurality rules here, i.e., we assume $1<k<m-1$.  

\end{tablenotes}
\end{threeparttable}
\caption{\emph{Summary of main results, with $n>2$, for tops-only, maxmin, scoring, negative plurality, $k$-approval, and Condorcet consistent rules.}}
\label{tabla familias}

\end{table}

\vspace{200 pt}

\bibliographystyle{ecta}
\bibliography{biblio-regret-free}

\appendix

\section{Appendix}\label{appendix}

\subsection{Proof of Theorem \ref{theo maxmin rules}} \label{proof theo maxmin}

We first show the equivalence in part (i). Let $f: \mathcal{P}^{n}\longrightarrow X$  be a  $A$-maxmin rule.

\noindent ($\Longrightarrow$) Assume that $%
n<m-1$ and  that $n$ does not divide  $m-1.$ Then, there are $h\geq 1$ and
$1\leq s<n$ such that $m-1=nh+s$, or  $m=nh+r$ with $h\geq 1$ and
$2\leq r\leq n.$ As $m-nh=r$, for any profile of preferences there are at least $r$ alternatives whose minimal position is at least $h+1$. So, the minimal position of a maxmin winner is always at least $h+1$, and in the case that it is exactly $h+1$, there are at least $r$ maxmin winners. Let $P \in \mathcal{P}^n$ be given by the following table:

\begin{center}
$%
\begin{array}{r cccccccc}
 & P_{1} & P_{2} & P_3 & \ldots & P_{r} & P_{r+1} & \ldots & P_n  \\ 
\cline{2-9}
 & b  & a  & b &\ldots & b & b &\ldots & b  \\
 & \vdots  &\vdots &\vdots &\ldots &\vdots &\vdots &\ldots &\vdots  \\
\cdashline{2-6} 
 & a & b &\vdots & \ldots &  
\multicolumn{1}{c:}{\vdots}
& \vdots & \ldots & \vdots \\
\cdashline{7-9}
\multirow{2}{*}{\parbox{15mm}{ \footnotesize{ \hspace*{5 pt} \text{  last $h$ }
\\ \text{ positions }} } $\begin{dcases*} \\ \\ 
\end{dcases*}$} & z &  \vdots &\vdots & \ldots &  \vdots &  \vdots & \ldots &  \vdots \\
 & \vdots  &\vdots &\vdots &\ldots &\vdots &\vdots &\ldots &\vdots  \\

\end{array}
$
\end{center}
where each alternative appears exactly one time below the dashed line (this can be done because $m=nh+r$).

Then, $a,b \in \mathcal{M}(P)$ and
$f(P)=a.$ 
Now, consider preference $P_{1}^{\prime
} \in \mathcal{P}$ that differs from $P_1$ only in that the positions of $a$ and $z$ are interchanged. We have that
$a \notin \mathcal{M}(P_{1}^{\prime
},P_{-1})$, $b \in \mathcal{M}(P_{1}^{\prime
},P_{-1})$ , and $f(P_{1}^{\prime
},P_{-1})=b$. Therefore,
\begin{equation}\label{newnew1}
f(P_{1}^{\prime
},P_{-1})P_{1}f(P_{1},P_{-1}).
\end{equation}
Let $P_{-1}^{\star } \in \mathcal{P}^{n-1}$  be such that
$f(P_{1},P_{-1}^{\star })=f(P)=a.$  As we noted in the first paragraph of this proof, 
it follows that $mp(f(P_{1}^{\prime },P_{-1}^{\star
}), (P_{1}^{\prime },P_{-1}^{\star }))\geq h+1$. There are two cases to consider:

\begin{enumerate}

\item[$\boldsymbol{1}$.] \textbf{$\boldsymbol{mp(f(P_{1}^{\prime
},P_{-1}^{\star }),(P_{1}^{\prime },P_{-1}^{\star
}))>h+1}$.} Then, $f(P_1',P_{-1}^\star) P_1'z$ and, by the definition of  $P_{1}^{\prime }$,
\begin{equation}\label{newnew2}
f(P_{1}^{\prime },P_{-1}^{\star })P_{1}f(P_{1},P_{-1}^{\star }).
\end{equation}
By \eqref{newnew1} and \eqref{newnew2}, $f$ is not regret-free truth-telling. 

\item[$\boldsymbol{2}$.] \textbf{$\boldsymbol{mp(f(P_{1}^{\prime
},P_{-1}^{\star }),(P_{1}^{\prime },P_{-1}^{\star
}))=h+1}$.} As we noted in the first paragraph of this proof,  $\left\vert
\mathcal{M}(P_{1}^{\prime },P_{-1}^{\star })\right\vert \geq r\geq 2$ and  $f(P_{1}^{\prime },P_{-1}^{\star })\neq z$ (because $z$
is the last one in  order $\succ $). Again, $f(P_1',P_{-1}^\star) P_1'z$ and the proof follows as in the previous case. 
\end{enumerate}

\noindent $(\Longleftarrow)$ Assume that there exist $i\in N,$ $(P_{i},P_{-i})\in \mathcal{P}^{n}$ and $%
P_{i}^{\prime }\in \mathcal{P}$ such that
\begin{equation}\label{0}
f(P_{i}^{\prime },P_{-i})P_{i}f(P_{i},P_{-i}). 
\end{equation}%
We will prove that there is $P_{-i}^{\star }\in \mathcal{P}^{n-1}$ such
that $f(P)=f(P_{i},P_{-i}^{\star })$ and $f(P_{i},P_{-i}^{\star
})P_{i}f(P_{i}^{\prime },P_{-i}^{\star }).$ Let $\widehat{P}=(P_{i}^{\prime },P_{-i}).$ As $f$ is an $A$-maxmin rule,
\begin{equation}
mp(P,f(P))\geq mp(P,f(\widehat{P}))  \label{1}
\end{equation}
and
\begin{equation}\label{2}
mp(\widehat{P},f(\widehat{P}))\geq mp(\widehat{P},f(P)).
\end{equation}%
Let $\overline{k}$ be such that $t_{\overline{k}}(P_i)=f(P).$  By (\ref{1}) and since $f(\widehat{P})P_{i}f(P),$
\begin{equation}\label{ktecho}
mp(P,f(\widehat{P}%
))=k^{\star }\leq \overline{k}\leq m-1,
\end{equation}
where $t_{k^{\star }}(\widehat{P}_{j})=f(\widehat{P})$ for some
$j\in N \setminus \{i\}.$ Then, as $\widehat{P}_{j}=P_{j}$ and $k^{\star }\leq \overline{k}$,
\begin{equation}
mp(P,f(\widehat{P}))\geq mp(\widehat{P},f(\widehat{P})).  \label{3}
\end{equation}%
If $mp(P,f(P))=mp(P,f(\widehat{P}))$ and $mp(\widehat{P},f(%
\widehat{P}))=mp(\widehat{P},f(P)),$ then
\begin{equation}
f(P),f(\widehat{P})\in \mathcal{M}(\widehat{P})\cap \mathcal{M}(P),
\label{33}
\end{equation}
contradicting 
that $f(P)\neq f(\widehat{P}).$ Therefore, by \eqref{1} and \eqref{2},
 $mp(\widehat{P},f(%
\widehat{P}))>mp(\widehat{P},f(P))$ or $mp(P,f(P))>mp(P,f(\widehat{P})).$ By (\ref{3}),
\begin{equation}
mp(P,f(P))>mp(\widehat{P},f(P)).  \label{4}
\end{equation}
Let $\widehat{k}$ be such that $mp(\widehat{P},f(P))=\widehat{k}.$ Then, by \eqref{4},  $t_{\widehat{k}}(\widehat{P}%
_{i})=f(P)$ and $f(P)\widehat{P}_{j}t_{\widehat{k}}(\widehat{P}_{j})$ for
all $j\in N \setminus \{i\}.$ If $\overline{k}\leq \widehat{k},$ then
\begin{equation*}
mp(\widehat{P},f(P))=\widehat{k}\geq \overline{k}\geq mp(P,f(P)),
\end{equation*}%
which contradicts (\ref{4}). Therefore,
\begin{equation}
\bar{k%
}>\widehat{k}\label{4a} .
\end{equation}
 This implies that there exists an alternative $x \in X$ such that
\begin{equation}
f(P)=t_{\bar{k}}(P_{i})P_{i}x  \text{  \ and  \ } 
xR_{i}^{\prime }t_{\bar{k}}(P_{i}^{\prime })P_{i}^{\prime }f(P).
\label{5}
\end{equation}
There are two cases to consider:
\begin{enumerate}

\item[$\boldsymbol{1}$.] $\boldsymbol{n\geq m-1}.$  Let  $P_{-i}^{\star } \in \mathcal{P}^{n-1}$ be such that $t(P_{j}^{\star })=f(P),$ $%
t_{m-1}(P_{j}^{\star })=x$ for each $j\in N\setminus \{ i\}$, and for each $x\in X \setminus \{f(P),x\}$
choose an agent $j^{x}$ such that $t_{1}(P_{j^{x}}^{\star })=x$ (this is
feasible because $n-1\geq m-2$). Now, let $P^{\star
}=(P_{i},P_{-i}^{\star }).$ Then, $mp((P_{i},P_{-i}^{\star }),y)=1$ for all $%
y\in X \setminus \{f(P),x\}$ and from definition of $%
P_{-i}^{\star },$ (\ref{5}) and the fact that $\overline{k}\leq m-1$, we have $mp((P_{i},P_{-i}^{\star }),f(P))=\bar{k}%
>mp((P_{i},P_{-i}^{\star }),x)$. Therefore, $f(P_{i},P_{-i}^{\star })=f(P).$
Furthermore, $mp((P_{i}^{\prime },P_{-i}^{\star }),y)=1$ for each $y\in
X \setminus \{f(P),x\}$ and from definition of $%
P_{-i}^{\star },$ (\ref{5}) and the fact that $\overline{k}\leq m-1$, we have $mp((P_{i}^{\prime },P_{-i}^{\star }),x)>mp((P_{i}^{\prime
},P_{-i}^{\star }),f(P))$. 
Therefore, $f(P_{i}^{\prime },P_{-i}^{\star })=c.$

We conclude that $f(P)=f(P_{i},P_{-i}^{\star })$ and, by (\ref{5}) and the fact that  $f(P_{i}^{\prime },P_{-i}^{\star })=x,$ $f(P_{i},P_{-i}^{\star
})P_{i}f(P_{i}^{\prime },P_{-i}^{\star }).$ Hence, $f$ is regret-free truth-telling. 

\item[$\boldsymbol{2}$.] $\boldsymbol{n$ \textbf{divides} $m-1}.$ Thus, $m-1=hn$ with $%
h\geq 1.$ Therefore, 
\begin{equation}
mp(P,f(P))\geq h+1. \label{37}
\end{equation}
Let $Y=\{y\in X:yP_{i}f(P)\}$. Then,
\begin{equation}
\left\vert Y \right\vert <m-mp(P,f(P))\leq m-(h+1)=hn+1-h-1=h(n-1). 
\label{36}
\end{equation}%
Let  $P_{-i}^{\star } \in \mathcal{P}^{n-1}$ be such that $t(P_{j}^{\star })=f(P),$ $%
t_{m-1}(P_{j}^{\star })=x$ for each $j \in N \setminus \{i\}$,  and for each $y\in Y$ 
choose  an agent $j$  and a position $u\leq h $ such that $t_{u}(P_{j}^{\star })=y$ (the construction of $P_{-i}^{\star }$ is feasible by (\ref{36}) and the fact that $m-2=hn-1\geq h(n-1)$). Now, let $P^{\star }=(P_{i},P_{-i}^{\star
}).$ Then, $mp((P_{i},P_{-i}^{\star }),y)\leq h$ for each $y\in Y,$ $mp((P_{i},P_{-i}^{\star }),f(P))\geq mp(P,f(P))\geq h+1$ (this holds by (\ref{37}) and the definition of $P_{-i}^{\star}$), and $mp((P_{i},P_{-i}^{\star }),f(P))=\bar{k}%
>mp((P_{i},P_{-i}^{\star }),y)$  for each  $y\in X \setminus Y$ (this follows from the definitions of $%
P_{-i}^{\star }$ and $Y$). Hence, $f(P_{i},P_{-i}^{\star })=f(P).$
Furthermore, $mp((P_{i}^{\prime },P_{-i}^{\star }),y)\leq h$ for each 
$y\in Y$ and $mp((P_{i}^{\prime },P_{-i}^{\star }),x)>mp((P_{i}^{\prime
},P_{-i}^{\star }),f(P))$ (this  follows from (\ref{ktecho}), (\ref{5}), and the definition of $%
P_{-i}^{\star }$). Therefore, $f(P_{i}^{\prime },P_{-i}^{\star }) \in X \setminus Y$ and  $f(P_{i}^{\prime },P_{-i}^{\star }) \neq f(P)$.

We conclude that $f(P)=f(P_{i},P_{-i}^{\star })$ and $f(P_{i},P_{-i}^{\star
})P_{i}f(P_{i}^{\prime },P_{-i}^{\star }).$  Hence, $f$ is regret-free truth-telling.

\end{enumerate}

Next, we show part (ii). Assume that $f:\mathcal{P}^n \longrightarrow X$ is a $N$-maxmin rule. Then, there exists $\overline{%
j} \in N$ such that \textit{\ }%
\begin{equation}\label{n}
f(\widetilde{P})=\underset{\widetilde{P}_{\overline{j}}}{\max } \ \mathcal{M}(%
\widetilde{P})\text{ for each }\widetilde{P} \in \mathcal{P}^n.
\end{equation}
Let $P,P_i',\widehat{P}, \overline{k},$ and $ \widehat{k}$  be as in  ($\Longleftarrow$) of part (i). It is easy to see that equations  \eqref{0}, \eqref{1}, \eqref{2}, \eqref{ktecho} and   \eqref{3} 
also hold here.

 If $mp(P,f(P))=mp(P,f(\widehat{P}))$ and $mp(\widehat{P},f(%
\widehat{P}))=mp(\widehat{P},f(P)),$ then \eqref{33} holds as in the proof of part (i). As  $f(P_{i}^{\prime
},P_{-i})P_{i}f(P_{i},P_{-i}),$ we have $\overline{j}\neq i$. But then (\ref{33}) contradicts  $f(P)\neq f(\widehat{P})$ since $P_{\overline{j}}=%
\widehat{P}_{\overline{j}}.$ Therefore, by \eqref{1} and \eqref{2}, 
 $mp(\widehat{P},f(%
\widehat{P}))>mp(\widehat{P},f(P))$ or $mp(P,f(P))>mp(P,f(\widehat{P})).$ Now, it is easy to see that equations \eqref{4}, \eqref{4a} and \eqref{5} hold in this proof as well, so there exists $x \in X$ such that   $f(P)=t_{\bar{k}}(P_{i})P_{i}x$ and
$xR_{i}^{\prime }t_{\bar{k}}(P_{i}^{\prime })P_{i}^{\prime }f(P).$

Now, we define profile $P_{-i}^{\star } \in \mathcal{P}^{n-1}$ where, for each $j \in N \setminus \{i\},$ $P_j^\star$ is differs from $P_j$ in that $f(P)$ is now in the top of $P_j^\star$ and $x$ is in the second place, while all the other alternatives keep their relative ranking.  Formally, let $P_{-i}^{\star } \in \mathcal{P}^{n-1}$ be such that, for each $j \in N \setminus \{i\},$ $t(P_{j}^{\star })=f(P),$ $%
t_{m-1}(P_{j}^{\star })=x,$ and if  $k^{\prime }$ and $k^{\prime \prime
}$ are such that $t_{k^{\prime }}(P_{j})=f(P)$ and $t_{k^{\prime
\prime }}(P_{j})=x$, if we let $k_{1}=\max \{k^{\prime },k^{\prime
\prime }\}$ and  $k_{2}=\min \{k^{\prime },k^{\prime \prime }\}$, define 
\begin{equation*}
t_{k}(P_{j}^{\star })=\left\{
\begin{tabular}{l}
$t_{k+2}(P_{j})$ \  \ \ if $m-2\geq k\geq k_{1}-1,$ \\
\\
$t_{k+1}(P_{j})$  \  \ \ if $\bar{k}-1>k\geq k_{2}.$%
\end{tabular}%
\right. 
\end{equation*}

Next, we present two claims.

\noindent \textbf{Claim 1:} $\boldsymbol{f(P_{i},P_{-i}^{\star })=f(P)}.$ Let $P^{\star }=(P_{i},P_{-i}^{\star }).$ Since $%
f(P)P_{i}x$ and by definition of $P_{-i}^{\star },$%
\begin{equation}\label{nueva1bis}
mp(P^{\star },f(P))>mp(P^{\star },x)
\end{equation}
Then, $f(P^{\star })\neq x.$ As $f(P_{i}^{\prime },P_{-i})P_{i}f(P)=t_{\overline{k}}(P_{i}),$
\begin{equation}
mp(P^{\star },f(P))=\overline{k}.
\end{equation}%
Now, let $b \in X\setminus \{f(P),x\}.$ By definition of $P^{\star }$
and the fact that $f$ is a $N$-maxmin rule,
\begin{equation}\label{eq222}
mp(P^{\star },b)\leq mp(P,b)\leq mp(P,f(P))\leq \overline{k}.
\end{equation}%
Therefore, $f(P)\in \mathcal{M}(P^{\star })$ and $mp(P^{\star },f(P))=\overline{k}.$ On the one hand, if  $\overline{j}\neq i,$ then  $t(P_{\overline{j}%
}^{\star })=f(P)$ and, by definition of $f,$ $f(P)=f(P^{\star }).$  On the other hand, if $\overline{j}=i$ and there is $b \in \mathcal{M}(P^{\star }) \setminus \{f(P)\}$, then
by \eqref{nueva1bis} and \eqref{eq222}, $mp(P,b)=mp(P,f(P))=\overline{k}.$
Thus, by (\ref{n}) and the fact that $\overline{j}=i,$ $f(P)P_{i}b.$
Therefore, as $P_{i}^{\star }=P_{i},$%
$f(P)=f(P^{\star }).$
 This proves the Claim. 

\medskip 
\noindent \textbf{Claim 2:} $\boldsymbol{f(P_{i},P_{-i}^{\star })P_{i}f(P_{i}^{\prime
},P_{-i}^{\star })}.$ If $f(P_{i}^{\prime },P_{-i}^{\star })=x,$ then by Claim 1 and \eqref{5} the proof is trivial. Now assume $f(P_{i}^{\prime },P_{-i}^{\star })\neq x.$ First, we will prove that $f(P_{i}^{\prime },P_{-i}^{\star })\neq
f(P).$ As $f(P)=t_{\widehat{k}}(P_{i}^{\prime }),$
\begin{equation*}
\widehat{k}=mp((P_{i}^{\prime },P_{-i}^{\star }),f(P)).
\end{equation*}%
Furthermore, as $xR_{i}^{\prime }t_{\overline{k}}(P_{i}^{\prime })$ and $\bar{k}\leq m-1,$ by
definition of $P_{-i}^{\star },$
\begin{equation}\label{40}
mp((P_{i}^{\prime },P_{-i}^{\star }),x)\geq \bar{k}.
\end{equation}
Then, by (\ref{4a}), 
\begin{equation*}
mp((P_{i}^{\prime },P_{-i}^{\star }),x)>\widehat{k}=mp((P_{i}^{\prime
},P_{-i}^{\star }),f(P)),
\end{equation*}%
implying that $f(P_{i}^{\prime },P_{-i}^{\star })\neq f(P).$ 

Now, let $b \in X \setminus \{f(P),x\}$ be such that  $bP_{i}f(P_{i},P_{-i}^{%
\star }).$ Since $f(P_{i},P_{-i}^{\star })=f(P)=t_{\bar{k}}(P_{i})$, by
definition of $f$ there exists $j\in N \setminus \{ i \}$ such that $t_{\bar{k}}(P_{j})R_{j}b$%
. By definition of $P_{-i}^{\star }$, $t_{\bar{k}}(P_{j}^{\star })R_{j}^{\star
}b$. Therefore,%
\begin{equation}\label{aa}
mp((P_{i}^{\prime },P_{-i}^{\star }),b)\leq \overline{k}.
\end{equation}
On the one hand, if $\overline{j}\neq i,$ since  $t_{m-1}(P_{\overline{j}%
}^{\star })=x$  the definition of $f,$  \eqref{40}, and \eqref{aa} imply that  $f(P_{i}^{\prime },P_{-i}^{\star })\neq b.$ 
On the other hand, if $\overline{j}=i,$ since $bP_{i}f(P_{i},P_{-i}^{\star })$ the definition of $f$ implies  $mp((P_{i},P_{-i}^{\star }),b)<mp((P_{i},P_{-i}^{\star
}),f((P_{i},P_{-i}^{\star }))).$ Then,
\begin{equation*}
mp((P_{i},P_{-i}^{\star }),b)<\overline{k}.
\end{equation*}%
Therefore, as $bP_{i}f(P_{i},P_{-i}^{\star })=t_{\bar{k}}(P_{i}),$
\begin{equation*}
mp((P_{i}^{\prime },P_{-i}^{\star }),b)<\overline{k}
\end{equation*}
Then, by the definition of $f$ and (\ref{40}), $f(P_{i}^{\prime
},P_{-i}^{\star })\neq b$ in this case as well. Therefore, we conclude that
\begin{equation*}
f(P_{i},P_{-i}^{\star })P_{i}f(P_{i}^{\prime },P_{-i}^{\star }),
\end{equation*} proving the Claim. 

By Claims 1 and 2 we conclude that $f$ is regret-free truth-telling. \hfill $\square$

\subsection{Proof of Theorem \ref{k star m-1}}\label{proof of theo scoring k star m-1}

 Let $f:\mathcal{P}^n \longrightarrow X$ be a scoring rule with  $s_{m-1}<s_{m}$. Let $a,b,c \in X$ and assume w.l.o.g. that  if $f$ is an $A$-scoring then the tie-breaking is given by order $\succ$ with $a\succ b\succ c \succ \ldots$, whereas if $f$ is a $N$-scoring rule agent $1$ break ties. There are two cases to consider:

\begin{enumerate}

\item[$\boldsymbol{1}.$] \textbf{$\boldsymbol{n=2t+3$ with $t\geq 0}$}. Let $P \in \mathcal{P}^n$ be given by the following table:
\begin{center}
$%
\begin{array}{ccc:ccc:ccc}
P_{1} & P_{2} & P_{3} & P_4 & \cdots & P_{t+3} & P_{t+4} & \cdots & P_{2t+3} \\
\hline

a & c & b & a & \cdots & a & b & \cdots & b \\
c & b & a & b & \cdots & b & a & \cdots & a \\
b & a & c & c & \cdots & c & c & \cdots & c \\
\vdots & \vdots & \vdots & \vdots & \cdots & \vdots & \vdots & \cdots & \vdots\\
\multicolumn{3}{c}{} & \multicolumn{3}{c}{$\upbracefill$} &\multicolumn{3}{c}{$\upbracefill$}\\
\multicolumn{3}{c}{} & \multicolumn{3}{c}{t \text{ agents}} &\multicolumn{3}{c}{t \text{ agents}}\\
\end{array}
$
\end{center}
As $s(P,a)=s(P,b)\geq s(P,x)$ for each $x \in X\setminus\{a,b\}$, by the tie-breaking it follows that $f(P)=a$. Let $P_{2}^{\prime }\in \mathcal{P}$ be such that $P_2':b,c,a, \ldots$,  and let  $\widehat{P}%
=(P_{2}^{\prime },P_{-2}).$ As $s_{m-1}<s_{m}$, we have that
\begin{equation*}
s(\widehat{P},b)>s(P,b)=s(P,a)=s(\widehat{P},a)
\end{equation*}%
and
\begin{equation*}
s(\widehat{P},b)>s(P,b)\geq s(P,c)>s(\widehat{P},c).
\end{equation*}%
Therefore, 
\begin{equation}\label{nuevo1}
f(\widehat{P})=bP_{2}a=f(P).
\end{equation}
Next, consider $P_{-2}^{\star } \in \mathcal{P}^{n-1}$  such that $f(P_{2},P_{-2}^{\star
})=a.$ 
Then, $f(P_{2}^{\prime },P_{-2}^{\star })\in \{a,b\}$ because $s((P_{2},P_{-2}^{\star }),b)<s((P_{2}^{\prime
},P_{-2}^{\star }),b)$,  $s((P_{2},P_{-2}^{\star
}),c)>s((P_{2}^{\prime },P_{-2}^{\star }),c),$ and 
$%
s((P_{2},P_{-2}^{\star }),x)=s((P_{2}^{\prime },P_{-2}^{\star }),x)$ for each $%
x\in X \setminus \{b,c\}$.  Therefore, 
\begin{equation}\label{nuevo2}
f(P_{2}^{\prime },P_{-2}^{\star })R_{2}f(P_{2},P_{-2}^{\star })
\end{equation} 
By \eqref{nuevo1} and \eqref{nuevo2}, $f$ is not regret-free truth-telling.

\item[$\boldsymbol{2}.$] \textbf{$\boldsymbol{n=2t$ with $t\geq 2}$}. Let $P \in \mathcal{P}^n$ be given by the following table:
\begin{center}
$%
\begin{array}{ccc:ccc:ccc}
P_{1} & P_{2} & P_{3} & P_4 & \cdots & P_{t+1} & P_{t+2} & \cdots & P_{2t} \\
\hline

a & c & c & a & \cdots & a & b & \cdots & b \\
b & b & a & b & \cdots & b & a & \cdots & a \\
c & a & b & c & \cdots & c & c & \cdots & c \\
\vdots & \vdots & \vdots & \vdots & \cdots & \vdots & \vdots & \cdots & \vdots \\
\multicolumn{3}{c}{} & \multicolumn{3}{c}{$\upbracefill$} &\multicolumn{3}{c}{$\upbracefill$}\\
\multicolumn{3}{c}{} & \multicolumn{3}{c}{t-2 \text{ agents}} &\multicolumn{3}{c}{t-1 \text{ agents}}\\
\end{array}%
$
\end{center}
Then,  $f(P)\in \{a,b,c\}.$ Furthermore, as $s(P,a)=s(P,b)$,
$a\succ b$ and $aP_{1}b$, $f(P)\in \{a,c\}$. If $f(P)=a$,  then, $%
s(P,a)=s(P,b)\geq s(P,c)$ and we proceed as in Case 1. If  
$f(P)=c$, then $s(P,c)\geq s(P,a)=s(P,b).$ Consider agent $j$ such that $t+2 \leq j \leq 2t$ (i.e., $P_j:b,a,c,\ldots$) and let $P_j' \in \mathcal{P}$ be such that $P_{j}^{\prime }:a,b,c, \ldots$  and
$\widehat{P}=(P_{j}^{\prime
},P_{-j}).$ As $s_{m-1}<s_{m}$, $s(\widehat{P},a)\geq s(\widehat{P},c)$ and $s(\widehat{P}%
,a)>s(\widehat{P},b)$. Since  $a\succ c$ and $aP_{1}c$,
\begin{equation}\label{nuevo3}
f(\widehat{P})=aP_{j}c=f(P).
\end{equation}
Next, consider $P_{-j}^{\star } \in \mathcal{P}^{n-1}$  such that $f(P_{j},P_{-j}^{\star })=c.$
Then, $f(P_{j}^{\prime },P_{-j}^{\star })\in \{c,a\},$ because $s((P_{j},P_{-j}^{\star }),a)<s((P_{j}^{\prime
},P_{-j}^{\star }),a),$ $(s(P_{j},P_{-j}^{\star
}),b)>s((P_{j}^{\prime },P_{-j}^{\star }),b),$ and  $%
s((P_{j},P_{-j}^{\star }),x)=s((P_{j}^{\prime },P_{-j}^{\star }),x)$ for each $%
x\in X\setminus \{ b,a\}$.   Therefore,
\begin{equation}\label{nuevo4}
f(P_{j}^{\prime },P_{-j}^{\star })R_{j}f(P_{j},P_{-j}^{\star })
\end{equation}
By \eqref{nuevo3} and \eqref{nuevo4}, $f$ is not regret-free truth-telling. 
\end{enumerate} 
\hfill $\square$

\subsection{Proof of Theorem \ref{scoring k=1}}\label{proof of theo scoring k=1}

 We first show the equivalence in part (i). Let $f:\mathcal{P}^n \longrightarrow X$ 
be an $A$-negative plurality rule.

\noindent $(\Longrightarrow)$  Suppose that $n<m-1$ (this
implies $m>3$)$.$
Assume that $a,b$ are the first two  alternatives in the tie-breaking with $%
a\succ b$ and let $z$ the last alternative in the tie-breaking$.$
Let $P \in \mathcal{P}^n$ 
be such that  $t_3(P_i)=b,$ $t_2(P_i)=a,$ $t_1(P_i)=z,$ and  $t_m(P_j)=b,$ $t_{m-1}(P_j)=a,$ and $t_{m-2}(P_j)=z$ for each $j \in N \setminus \{i\}.$  Then, $f(P)=a.$ Now, let $P_{i}^{\prime } \in \mathcal{P}$ be such that
$t_{1}(P_{i}^{\prime })=a.$ Then, $f(P_{i}^{\prime },P_{-i})=b$ and, therefore,
\begin{equation}\label{noregret1}
f(P_{i}^{\prime },P_{-i})P_{i}f(P).
\end{equation}
Now, let $P_{-i}^{\star } \in \mathcal{P}^{n-1}$ be such that $f(P)=f(P_{i},P_{-i}^{\star }).$ As $%
n+1<m,$ $\left\vert \mathcal{S}(P_{i}^{\prime },P_{-i}^{\star
})\right\vert
\geq 2.$ Therefore, as $z$ is the last alternative in the order $\succ ,$ $%
f(P_{i}^{\prime },P_{-i}^{\star })\neq z$ and 
\begin{equation}\label{noregret2}
f(P_{i}^{\prime },P_{-i}^{\star })R_{i}f(P_{i},P_{-i}^{\star }).
\end{equation}
Hence, by \eqref{noregret1} and \eqref{noregret2}, $f$  is not regret-free truth-telling. 

\noindent $(\Longleftarrow)$ Assume  $n \geq m-1$ and  there exist $i\in N,$ $P \in \mathcal{P}^{n}$ and $%
P_{i}^{\prime }\in \mathcal{P}$ such that
\begin{equation}
f(P_{i}^{\prime },P_{-i})P_{i}f(P).  \label{ar}
\end{equation}%
Next, we show there is $P_{-i}^{\star }\in \mathcal{P}^{n-1}$ such
that
$f(P)=f(P_{i},P_{-i}^{\star })$
and 
$f(P_{i},P_{-i}^{\star })P_{i}f(P_{i}^{\prime },P_{-i}^{\star }).$
Let $\widehat{P}=(P_{i}^{\prime },P_{-i}).$ If $t_{1}(P_{i})=t_{1}(P_{i}^{%
\prime }),$ then $\mathcal{S}(P)=\mathcal{S}(\widehat{P}),$ contradicting the definition of  $f$ and the fact that $f(\widehat{P})\neq
f(P_{i},P_{-i}).$ Therefore,
\begin{equation}
t_{1}(P_{i})\neq t_{1}(P_{i}^{\prime }). \label{GGG}
\end{equation}
If $%
t_{1}(P_{i})=f(P),$ then as $t_{1}(P_{i})\neq t_{1}(P_{i}^{\prime }),$ $s(%
\widehat{P},f(P))>s(P,f(P)).$ By definition of $f,$ $s(P%
,f(P)) \geq s(P,x)$ for each $x \in X.$ Then, $s(\widehat{P}%
,f(P))> s(P,x)$ for each $x \in X \setminus \{f(P)\}.$
Now, as $\widehat{P}=(P_{i}^{\prime },P_{-i}),$ $%
s(\widehat{P},f(P))>s(\widehat{P},x)$ for each $x \in X \setminus\{ t_{1}(P_{i})\}.$
Therefore, as $t_{1}(P_{i})=f(P),$ $f(\widehat{P})=f(P)$  which
contradicts \eqref{ar}. 
Thus,
\begin{equation}
t_{1}(P_{i})\neq f(P). \label{GG}
\end{equation}
Furthermore,
\begin{align*}
s(\widehat{P},x)&=s(P,x) \text{ for each } x\notin
\{t_{1}(P_{i}),t_{1}(P_{i}^{\prime })\}, \\
s(\widehat{P},t_{1}(P_{i}))&=s(P,t_{1}(P_{i}))+1, \text{and} \\
s(\widehat{P},t_{1}(P_{i}^{\prime }))&=s(P,t_{1}(P_{i}^{\prime }))-1. %
\end{align*}%
Then, as $s(P,x)\leq s(P,f(P))$ for each $x \in X,$
\begin{equation*}
\mathcal{S}(\widehat{P})=\{t_{1}(P_{i})\}\text{ or }\mathcal{S}%
(P) \setminus \{t_{1}(P_{i}^{\prime })\}\subset
\mathcal{S}(\widehat{P})\subset \mathcal{S}(P)\cup
\{t_{1}(P_{i})\}.
\end{equation*}
Thus, by \eqref{ar},
\begin{equation}\label{G}
\mathcal{S}%
(P) \setminus \{t_{1}(P_{i}^{\prime })\}\subset
\mathcal{S}(\widehat{P})\subset \mathcal{S}(P)\cup
\{t_{1}(P_{i})\}.
\end{equation}
Next, we claim that
\begin{equation}\label{claim3}
t_{1}(P_{i}^{\prime })=f(P)
\end{equation}
holds. Assume otherwise  that $%
t_{1}(P_{i}^{\prime })\neq f(P).$ Then, by \eqref{G}  and the definition
of $f$, $f(\widehat{P})=f(P)$ or $f(\widehat{P})=t_{1}(P_{i}),$ which
contradicts that $f(\widehat{P})P_{i}f(P).$ Then, \eqref{claim3} holds. 

Now, let $P_{-i}^{\star }$ be such that $\{t_{1}(P_{j}^{\star }):j\neq
i\}=X \setminus \{f(P),t_{1}(P_{i})\}$ ($P_{-i}^{\star }$ exists because
$n \geq m-1$). As $t_{1}(P_{i})\neq f(P),
$ $\mathcal{S}(P_{i},P_{-i}^{\star })=\{f(P)\}$ and, therefore,
\begin{equation*}\label{parte1}
f(P)=f(P_{i},P_{-i}^{\star })
\end{equation*}
By \eqref{claim3}, $t_{1}(P_{i}^{\prime })=f(P)$. Then,
$\mathcal{S}(P_{i}^{\prime },P_{-i}^{\star })=\{t_{1}(P_{i})\}$ implying 
$f(P_{i}^{\prime },P_{-i}^{\star })=t_{1}(P_{i})$
and, therefore,
\begin{equation}\label{parte2}
f(P_{i},P_{-i}^{\star })P_{i}f(P_{i}^{\prime },P_{-i}^{\star }).
\end{equation}
By \eqref{ar} and \eqref{parte2}, $f$ is regret-free truth-telling.

\medskip

In order to see (ii), Assume that $f$ is a $N$-negative plurality.  Then, there exists $\overline{%
j}$ such that \textit{\ }%
\begin{equation}\label{ene}
f(\widetilde{P})=\underset{\widetilde{P}_{\overline{j}}}{\max }\mathcal{S}(%
\widetilde{P})\text{ for each }\widetilde{P} \in \mathcal{P}^n.
\end{equation}
Let $P,P_i',\widehat{P},$  be as in ($\Longleftarrow$) of part (i). By definition, \eqref{ar}
also holds here.

If $t_{1}(P_{i})=t_{1}(P_{i}^{%
\prime }),$ then $\mathcal{S}(P)=\mathcal{S}(\widehat{P}).$ As  $f(P_{i}^{\prime
},P_{-i})P_{i}f(P_{i},P_{-i}),$ we have $\overline{j}\in N \setminus \{ i\}$. Now,
$\mathcal{S}(P)=\mathcal{S}(\widehat{P})$ contradicts  $f(P)\neq f(\widehat{P})$ since $P_{\overline{j}}=%
\widehat{P}_{\overline{j}}.$ Therefore, \eqref{GGG} holds here and it follows that both \eqref{GG} and \eqref{G} hold as well. If $\overline{j}=i$, we get a contradiction with 
$f(\widehat{P})P_{i}f(P)$ and $\mathcal{S}(\widehat{P})\subset
\mathcal{S}(P)\cup \{t_{1}(P_{i})\},$ so $\overline{j}\neq i$.

Now, let $P_{-i}^{\star } \in \mathcal{P}^{n-1}$
be such that $t(P_{\overline{j}}^{\star })=t_{1}(P_{i})$, $t_{m-1}(P_{\overline{%
j}}^{\star })=f(P)$ and, for each $j \in N \setminus \{i, \overline{j}\}$,  $t(P_{j}^{\star })=f(P)$ and $t_{2}(P_{j}^{\star
})=t_{1}(P_{i}).$ Therefore, by (\ref{GG}), $f(P)\in
\mathcal{S}(P_{i},P_{-i}^{\star }) $ and $t_{1}(P_{i})\notin
\mathcal{S}(P_{i},P_{-i}^{\star }).$ By definition of $f$
and $P_{\overline{j}}^{\star }$ it follows that
\begin{equation*}\label{rg1}
f(P)=f(P_{i},P_{-i}^{\star }).
\end{equation*}
Then, by (\ref{GGG}),
$t_{1}(P_{i})\in \mathcal{S}(P_{i}^{\prime },P_{-i}^{\star }).$  By
definition of $f$ and $P_{\overline{j}}^{\star }$ we have
\begin{equation*}
f(P_{i}^{\prime },P_{-i}^{\star })=t_{1}(P_{i})
\end{equation*}
Therefore,
\begin{equation}\label{rg2}
f(P_{i},P_{-i}^{\star })P_{i}f(P_{i}^{\prime },P_{-i}^{\star }).
\end{equation}
By \eqref{ar} and \eqref{rg2}, $f$ is regret-free truth-telling. \hfill $\square$

\subsection{Proof of Theorem \ref{scoring positive}}\label{proof theo scoring positive}

Assume $n>2$ and  let $f:\mathcal{P}^n \longrightarrow X$ be a scoring rule such that $k^\star n<m.$ We first show the equivalence in part (i). Let further assume that $f$ is  an $A$-scoring rule. 

\noindent $(\Longrightarrow)$ Assume that $\mathit{k^{\star }}n<m-1,$ we will prove that
$f$ is not regret-free truth-telling. Assume that $a$ and $b$ are the first two alternatives in
the tie-breaking $\succ$ with $a\succ b$ and let $z$ the last alternative in the
tie-breaking$.$ First, notice that  for any profile of preferences there are at least two alternatives above position $k^\star$ in the preference of each agent.  So, there are at least two score winners (with score equal to $n s_m$). 

Let $P \in \mathcal{P}^n$ be such that $t(P_{i})=b,$ $%
t_{k^{\star }+1}(P_{i})=a,$ $t_{k^{\star }}(P_{i})=z$ and, for each $j \in N\setminus \{i\},$ $P_{j}:b,a,\ldots.$ Then, $f(P)=a.$ Now, consider preference $P_{1}^{\prime} \in \mathcal{P}$ that differs from $P_1$ only in that the positions of $a$ and $z$ are interchanged. Therefore, $f(P_{i}^{\prime },P_{-i})=b$ and
\begin{equation}\label{nuevo5}
f(P_{i}^{\prime },P_{-i})P_{i}f(P).
\end{equation}%
Now, let $P_{-i}^{\star} \in \mathcal{P}^{n-1}$  be such that $f(P_{i},P_{-i}^{\star })=f(P).$ As we noted in the first paragraph of this proof, $s(f(P_{i}^{\prime
},P_{-i}^{\star }),(P_{i}^{\prime
},P_{-i}^{\star }))=n s_m$ and $f(P_{i}^{\prime
},P_{-i}^{\star })\neq z$ (because $z$ is the last alternative in the tie-breaking). Then, $f(P_{i}^{\prime
},P_{-i}^{\star })P_i'z$ and, by the definition of $P_i'$,
\begin{equation}\label{nuevo6}
f(P_{i}^{\prime },P_{-i}^{\star })R_{i}f(P_{i},P_{-i}^{\star }).
\end{equation}
By \eqref{nuevo5} and \eqref{nuevo6}, $f$ is not regret-free truth-telling.

\noindent $(\Longleftarrow )$ Assume that $\mathit{k^{\star }}n=m-1.$ Then, for any profile  there is always an alternative with maximal score $n \cdot s_m$.   Thus, given  $P\in \mathcal{P}^{n},$ 
$s(f(P),P)=n\cdot s_{m} 
$ 
 and $f(P)P_{j}t_{k^{\star }}(P_{j})$ for each $j \in N.$

Let $P_i' \in \mathcal{P}$ be such that \begin{equation}\label{nuevo7}
f(P_{i}^{\prime },P_{-i})P_{i}f(P).
\end{equation}%
Then, $f(P_i',P_{-i})P_if(P)P_it_{k^\star}(P_i).$ By definition of $k^{\star },$
\begin{equation}\label{nueva1}
s(f(P_i',P_{-i}),P)\geq s(f(P_i',P_{-i}),(P_{i}^{\prime },P_{-i}))
\end{equation}%
Also, 
\begin{equation}\label{nueva2}
t_{k^{\star}}(P_{i}')R_{i}^{\prime }f(P).
\end{equation}
Otherwise, $f(P)P_{i}^{\prime }t_{k^{\star }}(P_{i}')$ implies
 $s(f(P),P)=s(f(P),(P_{i}^{\prime },P_{-i})).
$ 
By  \eqref{nueva1}, $s(f(P_i',P_{-i},P))=s(f(P),P)$ and $s(f(P),(P_i',P_{-i}))=s(f(P_i',P_{-i}),(P_i',P_{-i})),$ contradicting the definition of $f$ since $f(P_i',P_{-i}) \neq f(P)$. So \eqref{nueva2} holds. 

Therefore, there exists $x\in X$ such that $%
t_{k^{\star }}(P_{i})R_{i}x$ and $xP_{i}^{\prime }t_{k^{\star }}(P_{i}').$ As $\mathit{k^{\star }}n=m-1$ is equivalent to $(n-1)k^{\star
}=m-k^{\star }-1$, we can consider  $P_{-i}^{\star } \in \mathcal{P}^{n-1}$ such that the two following requirements hold: (i) $P_{j}^{\star }:x,f(P),\ldots$ for each $%
j\in N \setminus \{i\},$ and (ii) for each $y\in X\setminus \{f(P)\}$ such that $yP_{i}t_{k^{\star
}}(P_{i})$ there exist $j\in N\setminus \{i\}$ such that $t_{k^{\star
}}(P_{j}^{\star })R_{j}^{\star }y$. Therefore, since now $f(P)$ is the only alternative with score $n \cdot s_m$, $%
\mathcal{S}(P_{i},P_{-i}^{\star })=\{f(P)\}$ and
\begin{equation*}
f(P_{i},P_{-i}^{\star })=f(P).
\end{equation*}%
As $s(x,(P_{i}^{\prime },P_{-i}^{\star }))=n\cdot s_{m}$ and, by \eqref{nueva2} and the definition of $P_{-i}^\star$, $%
s(r,(P_{i}^{\prime },P_{-i}^{\star }))<n\cdot s_{m}$ for each $r$ such that  $%
rP_{i}t_{k^{\star }}(P_{i}),$ it follows that $t_{k^\star}(P_i)R_if(P_{i}^{\prime },P_{-i}^{\star })$ and we have
\begin{equation}\label{nuevo8}
f(P_{i},P_{-i}^{\star })P_{i}f(P_{i}^{\prime },P_{-i}^{\star }).
\end{equation}
By \eqref{nuevo7} and \eqref{nuevo8}, $f$ is regret-free truth-telling. 

\medskip

\noindent To see part (ii), let $f:\mathcal{P}^n \longrightarrow X$ be a $N$-scoring rule. Then, there exists $\overline{j}\in N$ such
that \textit{\ }%
\begin{equation*}
f(\widetilde{P})=\underset{\widetilde{P}_{\overline{j}}}{\max } \ \mathcal{S}(%
\widetilde{P})\text{ for each }\widetilde{P} \in \mathcal{P}^n.  
\end{equation*}
Let $P\in \mathcal{P}^{n}$. As we noted in the first paragraph of this proof, 
$s(f(P),P)=n\cdot s_{m}
$ 
 and $f(P)P_{j}t_{k^{\star }}(P_{j})$ for each $j \in N.$

Let $P_i' \in \mathcal{P}$ be such that \begin{equation}\label{nuevo9}
f(P_{i}^{\prime },P_{-i})P_{i}f(P).
\end{equation}%
Since, by the definition of $k^{\star },$  $s(f(P_i',P_{-i}),P)=s(f(P),P)=n\cdot s_{m},$ it follows that  $i \in N \setminus\{\overline{j}\}.$

Notice that, by the same arguments, \eqref{nueva1} and \eqref{nueva2} also hold here.  Therefore, $t_{k^{\star }}(P_{i})P_{i}'f(P)$ and  there exists $x\in X$ such that $%
t_{k^{\star }}(P_{i})P_{i}x$ and $xP_{i}^{^{\prime }}t_{k^{\star }}(P_{i}).$

Now, let $P_{-i}^{\star} \in \mathcal{P}^{n-1}$ be such that $P_{j}^{\star}:x,f(P),\ldots$ for each $j\in N\setminus\{i\}.$ As $%
s(f(P),(P_{i},P_{-i}^{\star }))=n\cdot s_{m}>s(x,(P_{i},P_{-i}^{\star }))$ and $j\neq i$,
$f(P_{i},P_{-i}^{\star })=f(P).
$
As $s(x,(P_{i}^{\prime },P_{-i}^{\star }))=n\cdot s_{m}$ and $j\neq i$,%
 $f(P_{i}^{\prime },P_{-i}^{\star })=x.
$ 
 Then,
 \begin{equation}\label{nuevo10}
 f(P_i, P_{-i}^{\star})P_if(P_i', P_{-i}^\star).
 \end{equation}
 By \eqref{nuevo9} and \eqref{nuevo10}, $f$ is regret-free truth-telling. \hfill $\square$

\subsection{Proof of Theorem \ref{scoring negative}}\label{proof theo scoring negative}

Assume $n>2$ and let $f:\mathcal{P}^n \longrightarrow X$ be a scoring rule such that $k^{\star }n\geq m$ and $s_{k^{\star}-1}=s_{k^{\star }}$ (this implies $k^\star>1$). If $k^{\star }=m-1,$ the result follows from Theorem \ref{k star m-1}, so assume that $k^{\star }<m-1.$ If $f$ is and $A$-scoring rule, assume that $a$ and $b$ are the first two alternatives in
the tie-breaking $\succ$ with $a\succ b$, whereas if $f$ is a $N$-scoring rule, let agent 1 be the one who break ties.
By the definition of $k^{\star
},$ $s_{k^{\star }-1}=s_{k^{\star
}}<s_{k^{\star }+1}=s_{m-1}=s_{m}$. Let $a,b \in X.$ As $\mathit{k^{\star }}n\geq m$ and $%
k^{\star }>1$,  $k^{\star }(n-1)\geq m-k^{\star }$.  Then,
there exists $P \in \mathcal{P}^n$ such that:

\begin{enumerate}[(i)]
\item $b=t_{k^{\star }}(P_{2})$  and $a=t_{k^{\star }-1}(P_{2})$,  

\item for each $j\in N\setminus \{2\},$ $t(P_{j})=a$ and $t_{m-1}(P_{j})=b,$

\item for each $x\in X$ such that $xP_{2}b$, there exist $j\in N\setminus
\{2\}$ such that $t_{k^{\star }}(P_{j})R_{j}x.$
\end{enumerate}
Since  $s(a,P)\geq s(x,P)$ for each $x \in X$ such that $xR_{2}b,$ $a\succ x$ and $%
aP_{1}b$, it follows that  $bP_{2}f(P)$. Let $P_{2}^{\prime } \in \mathcal{P}$ be such that $t_{k^{\star
}+1}(P_{2}^{\prime })=b=t_{k^{\star }}(P_{2}),$ $t_{k^{\star }}(P_{2}^{\prime
})=t_{k^{\star }+1}(P_{2})$, and $t_{k}(P_{2}^{\prime })=t_{k}(P_{2})$ for
each $k\neq k^{\star },k^{\star }+1$. Let $\widehat{P}=(P_{2}^{\prime
},P_{-2}).$ Then, by the definition of $k^{\star },$
\begin{equation*}
s(b,\widehat{P})>s(b,P)=s(a,P)=s(a,\widehat{P})
\end{equation*}%
and
$s(b,\widehat{P})\geq s(x,\widehat{P})
$ 
   if $bP_2x.$ Therefore, 
\begin{equation}\label{nuevo11}
f(P_2',P_{-2})P_{2}f(P).
\end{equation}   
Let $P_{-2}^{\star } \in \mathcal{P}^{n-1}$ be such that $f(P_{2},P_{-2}^{\star })=f(P).$ Since  $s(f(P), (P_{2},P_{-2}^{\star }))=s(f(P),(P_{2}^{\prime
},P_{-2}^{\star }))$ and $s(x,(P_{2},P_{-2}^{\star }))\geq s(x,(P_{2}^{\prime },P_{-2}^{\star }))$ for
each $x\in X \setminus \{b\},$ it follows that $f(P_{2}^{\prime },P_{-2}^{\star })\in \{f(P),b\}$. Therefore,
\begin{equation}\label{nuevo12}
f(P_{2}^{\prime },P_{-2}^{\star })R_{2}f(P_{2},P_{-2}^{\star }).
\end{equation}
By \eqref{nuevo11} and \eqref{nuevo12}, $f$ is not regret-free truth-telling. \hfill $\square$

\subsection{Proof of Theorem \ref{theo approval} }\label{proof of theo approval}

Observe that in a $k$-approval voting rule with $1<k<m-1$, we have $1<k^\star=m-k<m-1$ and $s_{k^\star}=s_{k^\star-1}$. Therefore, by Theorems \ref{scoring positive} and \ref{scoring negative}, we have that:
\begin{itemize}
    \item an anonymous $k$-approval voting rule is regret-free truth-telling if and only if $(m-k)n=m-1$  $\left( \text{or }k=\frac{nm-m+1}{n}\right)$.
    \item a neutral $k$-approval voting rule is regret-free truth-telling if and only if  $(m-k)n<m$ $\left( \text{or } k>\frac{m(n-1)}{n} \right).$ \hfill $\square$
\end{itemize}

\subsection{Proof of Theorem \ref{theo Condorcet consistent} }\label{proof theo Condorcet consistent}

 Let $f: \mathcal{P}^n \longrightarrow X$ be a Condorcet consistent and monotone rule. There are two cases to consider:
\begin{enumerate}

\item[$\boldsymbol{1}.$] $\boldsymbol{n \neq 2,4}.$ Then, there are $t \geq 1$ and $s \geq 0$ such that $n=3t+2s.$ Let $P \in \mathcal{P}^n$ be given by the following table:
\begin{center}
$%
\begin{array}{ccc:ccc:ccc}
P_{1} & \cdots & P_t & P_{t+1} & \cdots & P_{2t+s} & P_{2t+s+1} & \cdots & P_{3t+2s} \\
\hline

a & \cdots & a & b & \cdots & b & c & \cdots & c \\
b & \cdots & b & c & \cdots & c & a & \cdots & a \\
c & \cdots & c & a & \cdots & a & b & \cdots & b \\
\vdots & \cdots & \vdots & \vdots & \cdots & \vdots & \vdots & \cdots & \vdots\\ 
\multicolumn{3}{c}{$\upbracefill$} & \multicolumn{3}{c}{$\upbracefill$} &\multicolumn{3}{c}{$\upbracefill$}\\
\multicolumn{3}{c}{t \text{ agents}} & \multicolumn{3}{c}{t+s \text{ agents}} &\multicolumn{3}{c}{t+s \text{ agents}}\\
\end{array}
$
\end{center}
Since $C_{P}(a,c)=t<\frac{3t+2s}{2},$ $C_{P}(c,b)=t+s<\frac{3t+2s}{2},$ and $%
C_{P}(b,a)=t+s<\frac{3t+2s}{2},$ it follows that there is no Condorcet winner according to $P.$

Let $x=f(P)$. Then, there exists $i^\star \in N$ such that $x=t_{k}
(P_{i^\star})$ with $k\leq m-2$.
Assume first that  $i^\star$ is such that $t+1 \leq i^\star \leq 2t+s$. Let $N'=\{j \in N : t+1 \leq j \leq 2t+s\}$ and consider the subprofile $P'_{N'} \in \mathcal{P}^{t+s}$ where, for each $j \in N',$ $P_j' \in \mathcal{P}$ is such that $t(P_j')=c,$ $t_{m-1}(P_j')=b,$ $t_{m-2}(P_j')=a,$ and $%
t_{k}(P_{j}')=t_{k}(P_{j})$ for each $k\leq m-3.$ Then, $c$ is the
Condorcet winner in $(P_{N'}^{\prime },P_{-N'}).$ As $i^\star \in N',$ $x\neq c.$ This implies the existence of  $S\subset N'$ and $j^\star \in N'\setminus S$ such that
\begin{equation}\label{nuevo13}
f(P_{S}^{\prime },P_{-S})=x
\end{equation}
and
\begin{equation}\label{nuevo14}
f(P_{S\cup \{j^\star\}}^{\prime
},P_{-S\cup \{j^\star\}})\neq x.
\end{equation}
Now, by monotonicity and \eqref{nuevo14}, 
$f(P_{S\cup \{j^\star\}}^{\prime
},P_{-S\cup \{j^\star\}})P_{j^\star} x,$ 
implying 
\begin{equation}\label{nuevo15}
f(P_{S\cup \{j^\star\}}^{\prime },P_{-S\cup \{j^\star\}})P_{j^\star}f(P_{S}^{\prime
},P_{-S}).
\end{equation}
Now let, $P_{-j^\star}^{\star } \in \mathcal{P}^{n-1}$ be such that $f(P_{j^\star},P_{-j^\star}^{\star })=f(P_{S}^{\prime
},P_{-S}).$ By \eqref{nuevo13}, $f(P_{j^\star},P_{-j^\star}^{\star })=x$. Then, by monotonicity, $f(P_{j^\star}^{\prime },P_{-j^\star}^{\star })R_{j^\star}x.$ Hence  
\begin{equation}\label{nuevo16}
f(P_{j^\star}^{\prime },P_{-j^\star}^{\star })R_{j^\star}f(P_{j^\star},P_{-j^\star}^{\star }).
\end{equation}
By \eqref{nuevo15} and \eqref{nuevo16}, $f$ is not regret-free truth-telling. The cases where  $i^\star$ is such that $1 \leq i^\star \leq t$  or $2t+s+1 \leq i^\star \leq 3t+2s$ are similar and therefore we  omit them. 

\item[$\boldsymbol{2}.$] $\boldsymbol{n = 4$ \textbf{and} $m>3}.$  Let $P \in \mathcal{P}^n$ be given by the following table:

\begin{center}
$
\begin{array}{cccc}
P_1 & P_2 & P_3 & P_4 \\
\hline
a & b & c & d \\
b & c & d & a \\
c& d & a & b \\
d & a & b & c \\
\vdots & \vdots & \vdots & \vdots\\
\end{array}
$
\end{center}

As $C_{P}(a,c)=C_{P}(a,b)=C_{P}(b,d)=2,$ there is no Condorcet
winner according to  $P.$ Let $f(P)=x.$ Assume that $f(P)\notin \{b,c,d\}$ (the
other 3 cases in which $f(P)\notin \{w,u,h\}$ with
$\{w,u,h\}\subset\{a,b,c,d\}$ follow a similar argument). Next, let $P_2' \in \mathcal{P}$ be  such that $t(P_2')=d,$ $t_{m-1}(P_2')=b,$ $t_{m-2}(P_2')=c,$ $t_{m-3}(P_2')=a,$ and $%
t_{k}(P_2')=t_{k}(P_{2})$ for each $k\leq m-4.$ Similarly, let $P_3' \in \mathcal{P}$ be  such that $t(P_3')=d,$ $t_{m-1}(P_3')=c,$ $t_{m-2}(P_3')=a,$ $t_{m-3}(P_3')=b,$ and $%
t_{k}(P_3')=t_{k}(P_{3})$ for each $k\leq m-4.$ Then, $d$ is the Condorcet winner according to $(P_{\{2,3\}}^{\prime
},P_{-\{2,3\}}).$ There are two cases to consider:
\begin{enumerate}

\item[$\boldsymbol{2.1}.$] $\boldsymbol{f(P_2', P_{-2})\neq x}.$ Then, by monotonicity, $f(P_{2}^{\prime }, P_{-2})P_2x.$ Hence, 
\begin{equation}\label{nuevo17}
f(P_{2}^{\prime }, P_{-2})P_2f(P).
\end{equation} 
\end{enumerate}
Now, let $P_{-2}^{\star } \in \mathcal{P}^{n-1}$ be such that $f(P_{2},P_{-2}^{\star })=f(P).$ 
Then, by monotonicity,
\begin{equation}\label{nuevo18}
f(P_{2}^{\prime },P_{-2}^{\star })R_{2}f(P_{2},P_{-2}^{\star }).
\end{equation}
By \eqref{nuevo17} and \eqref{nuevo18}, $f$ is not regret-free truth-telling.

\item[$\boldsymbol{2.2}.$] $\boldsymbol{f(P_2', P_{-2})=x}.$ Then,  $f(P_{\{2,3\}}^{\prime },P_{-\{2,3\}})=dP_3x=f(P_2',P_{-2})$ and an analogous reasoning to the one presented in Case 2.1 for agent 2, now performed with agent 3, shows that $f$ is not regret-free truth-telling. \hfill $\square$
\end{enumerate}

\subsection{Proof of Corollary \ref{corollary Simpson et al}}\label{proof of corollary Simpson et al}

We first show that each of the rules is monotone. 

\begin{lemma}\label{lema Simpson et al} Simpson, Copeland, Young, Dodgson, Fishburn and Black rules (both anonymous and neutral) satisfy monotonicity.
\end{lemma}
\begin{proof} Let $x\in X,$ $P \in \mathcal{P}^n$ and  $P_i' \in \mathcal{P}$ be such that $P_i'$ is a monotonic transformation of $P_i$ with respect to $x.$ Let $z \in X$ be such that $xP_iz$ and let $y \in \{x,z\}.$ Then, $C_P(y,a)=C_{(P'_i, P_{-i})}(y,a)$ for each $a \in X.$ Therefore, (both anonymous and neutral) Simpson, Copeland and Fishburn rules are monotonic. To see that Young and Dodgson rules are monotonic, simply note that $yP_i a$ if and only if $yP_i'a$ for each $a \in X\setminus\{y\}$. Finally, to see that Black rule is monotonic, note that (i) $y$  is a Cordorcet winner in $P$ if and only if $y$ is a Condorcet winner in $(P'_i, P_{-i}),$ and (ii) the Borda score for $y$ is the same in profiles $P$ and $(P'_i, P_{-i}).$
\end{proof}

\bigskip

\noindent \emph{Proof of Corollary \ref{corollary Simpson et al}}. Assume first that $N=\{1,2,3,4\}$ and $X=\{a,b,c\}.$  In all of the cases that we consider in what follows,
w.l.o.g., we assume  that the tie-breaking is given by $a \succ b \succ c$ in the anonymous case, or by agent 1 in the neutral case.

Let $f:\mathcal{P}^{4}\longrightarrow {\{a,b,c\}}$ be a Simpson (Young, Dodgson, Fishburn) rule.  Le $P \in \mathcal{P}^{4}$ be given by the following table: \begin{center}
$%
\begin{array}{cccc}
P_1 & P_2 & P_3 & P_4 \\
\hline
b & c & c & a \\
a & b & b & c \\
c& a & a & b \\

\end{array}%
$
\end{center}
Then, $c$ is the only Simpson (Young, Dodgson,  Fishburn) winner at $P$ and $f(P)=c.$ Now, consider  $P_1' \in \mathcal{P}$ such that $P_1': a,b,c.$ Then, $a$ is a Simpson (Young, Dodgson, Fishburn) winner at $(P_1',P_{-1}).$ Therefore, $f(P_1', P_{-1})=aP_1c=f(P).$ Let  $P_{-1}^\star \in \mathcal{P}^{n-1}$ be such that $f(P_1, P_{-1}^\star)=f(P).$ Since $f(P)=c=t_1(P_1)$, $f(P_1', P_{-1}^\star)R_1f(P_1, P_{-1}^\star)$. Hence, $f$ is not regret-free truth-telling.   

Next, let
$f:\mathcal{P}^{4}\longrightarrow {\{a,b,c\}}$ be a Copeland (Black) rule.  Let $P \in \mathcal{P}^{4}$ be given by
\begin{center}
$%
\begin{array}{cccc}
P_1 & P_2 & P_3 & P_4 \\
\hline
b & c & c & a \\
a & a & b & c \\
c& b & a & b \\

\end{array}%
$
\end{center}
Then, $c$ is the only Copeland (Black) winner at $P$ and $f(P)=c.$ Now, consider $P_1' \in \mathcal{P}$ such that $P_1': a,b,c.$ Then, $a$ is a Copeland (Black) winner at $(P_1', P_{-1})$ and a similar reasoning to the one presented for Simpson' rule shows that $f$ is not regret-free truth-telling. 

Finally, assume $n\notin\{2,4\}$, or $n=4$ and $m>3.$ By Lemma \ref{lema Simpson et al}, Simpson, Copeland, Young, Dodgson, Fishburn and Black rules are monotonic. Since all of them are also Condorcet consistent, the result follows from Theorem \ref{theo Condorcet consistent}. 
\hfill $\square$

\subsection{Proof of Theorem \ref{theo succesive} }\label{proof theo succesive}

Let $f:\mathcal{P}^n \longrightarrow X$ be a successive elimination rule with associated order $a \succ b \succ c \succ \ldots$ and let $t \geq 1$ and $1\geq s\geq 0$ be such that $n=2t+s$. Next, let $P \in \mathcal{P}^n$ be given by the following table:\footnote{Notice that, as $n\geq 3,$ $t+s-2\geq 0.$}
\begin{center}
$%
\begin{array}{cc:ccc:ccc}
P_{1} & P_2 & P_{3} & \cdots & P_{t+2} & P_{t+3} & \cdots & P_{2t+s} \\
\hline

a  & c & b & \cdots & b & a & \cdots & a \\
b  & a & c & \cdots & c & b & \cdots & b \\
\vdots & b & a & \cdots & a & c & \cdots & c \\
\vdots & \vdots & \vdots & \cdots & \vdots & \vdots & \cdots & \vdots\\
c  & \vdots & \vdots & \cdots & \vdots & \vdots & \cdots & \vdots\\
\multicolumn{2}{c}{} & \multicolumn{3}{c}{$\upbracefill$} &\multicolumn{3}{c}{$\upbracefill$}\\
\multicolumn{2}{c}{} & \multicolumn{3}{c}{t \text{ agents}} &\multicolumn{3}{c}{t+s-2 \text{ agents}}\\
\end{array}
$
\end{center}

\medskip

\noindent Since  $C_{P}(a,b)=t+s\geq t=C_{P}(b,a)$, $C_{P}(a,c)=t+s-1<t+1=C_{P}(c,a),$ and $C_{P}(c,x)=n-1>1=C_{P}(x,c)$ for each $x \in X \setminus\{a,b\}$, it follows that 
 $f(P)=c$.
Let $P_1' \in \mathcal{P}$ be such that $t(P_1')=b$, $t_{m-1}(P_1')=a,$ and $t_1(P_1')=c$, and let $\widehat{P}=(P_1', P_{-1})$.  Since 
$C_{\widehat{P}}(a,b)=t+s-1<t+1=C_{\widehat{P}}(b,a),$
$C_{\widehat{P}}(b,c)=n-1>1=C_{\widehat{P}}(c,b)$, and
$C_{\widehat{P}}(b,x)>C_{\widehat{P}}(x,b)$ for each $x\in X \setminus \{a,c\}$, it follows that 
 $f(P_{1}^{\prime },P_{-1})=b.
$ 
 Therefore, 
\begin{equation}\label{nuevo19} 
f(P_{1}^{\prime },P_{-1})P_{1}f(P).
\end{equation}
Furthermore, as $f(P)=t_{1}(P_{1}),$
\begin{equation}\label{nuevo20} 
f(P_{1}^{\prime },P_{-1}^{\star })R_{1}f(P_{1},P_{-1}^{\star })
\end{equation}%
for each $P_{-1}^{\star} \in \mathcal{P}^{n-1}$ such that $f(P_{1},P_{-1}^{\star })=f(P).$ By \eqref{nuevo19} and \eqref{nuevo20}, $f$ is not regret-free truth-telling. \hfill $\square$

\subsection{Proof of Theorem \ref{theo neutral 2x3}}\label{proof theo neutral 2x3}

$(\Longrightarrow )$ Let $f:\mathcal{P}^2 \longrightarrow \{a,b,c\}$ be a regret-free truth-telling and neutral  rule.

\noindent \textbf{Claim: $\boldsymbol{f}$ is efficient}. Assume $f$ is not efficient. W.l.o.g., there are two cases to consider:

\begin{enumerate}

\item[$\boldsymbol{1}.$] \textbf{$\boldsymbol{P \in \mathcal{P}^2$ is such that $f(P)=c$, $P_i: a,b,c}$,  and $\boldsymbol{P_j$ is such that $x=t(P_j)\neq c}.$} By regret-free truth-telling, $f(P_i', P_j)=c$ for each $P_i' \in \mathcal{P}.$ Let $\pi$ be the permutation of $X$ such that $\pi(c)=x.$ By neutrality,  $f(\pi P)=x.$ Then, by regret-free truth-telling, $f(P_i', \pi P_j)=x$ for each $P_i' \in \mathcal{P}.$ This implies that, as $f(\pi P)=xP_j c= f(\pi P_i, P_j)$ and  $f(P_i', \pi P_j)=x$ for each $P_i' \in \mathcal{P}$, agent $j$ manipulates $f$ and does not regret it.

\item[$\boldsymbol{2}.$] \textbf{$\boldsymbol{P \in \mathcal{P}^2$ is such that $f(P)=b$ and $P_i=P_j: a,b,c}$.}  Let $\pi$ be the permutation of $X$ such that $\pi(a)=b.$ By neutrality,  $f(\pi P)=a.$ By the previous case, $f(\pi P_i, P_j)\neq c.$ We claim  that $f(\pi P_i, P_j)=b$.  Assume $f(\pi P_i, P_j)=a.$ Let $P_j^\star$ be such that $f(P_i, P_j^\star)=b.$ If   $f(\pi P_i, P_j^\star)=c$, then agent $i$ manipulates $f$ at $(\pi P_i, P_j^\star)$ via $P_i$ and does not regret it. Therefore,  $f(\pi P_i, P_j^\star)\neq c.$ This implies that agent $i$ manipulates $f$ at $P$ via $\pi P_i$ and does not regret it. This proves the claim that $f(\pi P_i, P_j)=b$. By a similar reasoning to the one presented for agent $i$, we can see that agent $j$ manipulates $f$ at $(\pi P_i,  P_j)$ via $\pi P_j$ and does not regret it.

\end{enumerate}
Since in both cases we reach a contradiction, $f$ is efficient. This proves the claim.

Next, assume that $f$ is not a dictatorship. We will prove that $f$ is a $N$-maxmin rule.  Let $\overline{P} \in \mathcal{P}^2$ be such that $\overline{P}_{1}: a,b,c$ and $\overline{P}_{2}: b,a,c$.  By
efficiency,  $f(\overline{P})\in \{a,b\}.$ Assume, w.l.o.g., that $f(\overline{P})=a$.
We will prove that
\begin{equation*}
f(P)=\max_{P_{1}} \ \mathcal{M}(P)  \ \text{ for each }P \in \mathcal{P}^2.
\end{equation*}
Let $P\in \mathcal{P}^2.$ There are three cases to consider:

\begin{enumerate}

\item[$\boldsymbol{1}$.]  $\boldsymbol{t(P_{1})=t(P_{2})}.$  By efficiency,
$f(P)=t(P_{1})=\max_{P_{1}}\ \mathcal{M}(P).$

\item[$\boldsymbol{2}$.] \textbf{$\boldsymbol{t(P_{1})\neq t(P_{2})$ and $%
t_{1}(P_{1})=t_{1}(P_{2})}$.}  As $f(\overline{P})=a$, by neutrality, $f(P)=t(P_{1})=\max_{P_{1}} \ \mathcal{M}(P).$

\item[$\boldsymbol{3}$.] \textbf{$\boldsymbol{t(P_{1})\neq t(P_{2})$ and $t_{1}(P_{1})\neq
t_{1}(P_{2})}$.} Then,
\begin{equation}\label{Eme}
\mathcal{M}(P)=X \setminus \{t_1(P_1), t_1(P_2)\}.
\end{equation}

If $f(P)=t_{1}(P_{i})=x$ for some $i\in \{1,2\},$ then $f(P)=t(P_{j})=x$
with $j\neq i$ (because of efficiency). Then, by regret-free truth-telling,
\begin{equation*}
f(P_{j},P_{i}^{\prime })=x\text{ for all }P_{i}^{\prime}.
\end{equation*}
Then, again by regret-free truth-telling,
\begin{equation*}
f(P_{j}^{\prime },P_{i}^{\prime })=x\text{ for all }P_{i}^{\prime }\text{
and all }P_{j}^{\prime }\text{ such that }t(P_{j}^{\prime })=x.
\end{equation*}
Then, $j$ is a dictator when he has top in $x.$ Therefore, by neutrality, $j$ is a dictator which is a contradiction. Thus,
\begin{equation}\label{fnot}
f(P)\neq t_{1}(P_{i})\text{ for all }i\in \{1,2\}.
\end{equation}
Therefore, by \eqref{Eme} and \eqref{fnot}, $\mathcal{M}(P)=\{f(P)\}$ and
$f(P)=\max_{P_{1}}\ \mathcal{M}(P).$

\end{enumerate}

\noindent ($\Longleftarrow )$ Let $f$ be a $N$-maxmin rule. It is clear that $f$ is neutral and, furthermore, by Theorem \ref{theo maxmin rules} (ii), $f$ is regret-free truth-telling. If $f$ is a dictatorship, it is trivial that it is neutral and regret-free truth-telling. \hfill $\square$

\subsection{Proof of Theorem \ref{theo eff anon 2x3}}\label{proof theo eff anon 2x3}

\noindent $(\Longleftarrow)$ Let $f:\mathcal{P}^2 \longrightarrow \{a,b,c\}$ be a successive elimination rule or an $A$-maxmin$^{\star }$ rule. It is clear that $%
f$ satisfies efficiency and anonymity. We will prove that $f$ is
regret-free truth-telling. Assume there are $(P_{1},P_{2})\in \mathcal{P}^{2}$ and $%
P_{1}^{\prime }\in \mathcal{P}$ such that
\begin{equation}
f(P_{1}^{\prime },P_{2})P_{1}f(P_{1},P_{2}).   \label{eee}
\end{equation}%
We will prove that there exists $P_{2}^{\star }\in \mathcal{P}$ such that
$f(P_{1},P_{2}^{\star })=f(P)$
and
\begin{equation}\label{new3}
f(P_{1},P_{2}^{\star })P_{1}f(P_{1}^{\prime },P_{2}^{\star }).
\end{equation}%
There are two cases to consider:

\begin{enumerate}

\item[$\boldsymbol{1}.$] \textbf{$\boldsymbol{f}$ is a successive elimination rule with associated  order $\boldsymbol{a\succ
b\succ c}$}. It is clear that
\begin{equation}
f(\widetilde{P})\widetilde{R}_{i}a\text{ for each }\widetilde{P} \in \mathcal{P}^2\text{ and each }i \in \{1,2\}.  \label{ee}
\end{equation}
If $f(P_1, P_2)=a,$ by \eqref{eee} and efficiency, $aP_2 f(P_1', P_2),$ contradicting \eqref{ee}. Therefore, $f(P_1,P_2) P_1 a$ and, by \eqref{eee}, $t_1(P_1)=a.$
There are two cases to consider:

\begin{enumerate}

\item[$\boldsymbol{1.1}.$] \textbf{$\boldsymbol{b P_1 c P_1 a}$.} Then, by \eqref{eee}, $f(P_1,P_2)=c$ and, by definition of $f,$ $cP_2aP_2b.$  Therefore, there is no $P_1' \in \mathcal{P}$ such that $f(P_1', P_2)=b,$ contradicting \eqref{eee}.

\item[$\boldsymbol{1.2}.$] \textbf{$\boldsymbol{c P_1 b P_1 a}$.} Then, by \eqref{eee}, $f(P_1,P_2)=b$ and, by definition of $f,$ $t(P_2)=b.$ It follows from \eqref{eee} that $f(P_1', P_2)=c,$ implying  that $cP_2a$ and $cP_1'aP_1'b.$ Now, let $P_{2}^{\star }\in \mathcal{P}$ be such that $bP_{2}^{\star }aP_{2}^{\star }c.$ Then,  $f(P_{1},P_{2}^{\star })=b$
and $f(P_{1}^{\prime },P_{2}^{\star })=a.$ Since  $bP_{1}a,$ \eqref{new3} holds and $f$ is regret-free truth-telling. 

\end{enumerate}

\item[$\boldsymbol{2}.$]  \textbf{$\boldsymbol{f$ is a A-maxmin$^{\star }$ rule with associated  binary
relation $\succ ^{\star }}$}. By definition of $f,$ it is clear that
\begin{equation}
f(\widetilde{P})\neq t_{1}(\widetilde{P}_{i}) \text{ for each }\widetilde{P} \in \mathcal{P}^2\text{ and each }i \in \{1,2\}.  \label{eeee}
\end{equation}
W.l.o.g, let $P_1:a,b,c.$ By (\ref{eee}), $t(P_{2})\neq a$ and $f(P)\neq a.$ Then, by (%
\ref{eeee}), $f(P)=b \in \{t(P_{2}),t_{2}(P_{2})\}.$ By \eqref{eee} and \eqref{eeee}, $f(P_1', P_2)=a \in \{t(P_2), t_2(P_2)\}.$ Therefore, as $t(P_2)\neq a,$  $P_2: b,a,c.$ Then, by definition of $f,$ $b \succ ^{\star }a.$ Therefore,  as  $f(P_{1}^{\prime },P_{2})=a,$ $t_1(P_1')=b.$ Now let $P_{2}^{\star }:b,c,a.$ Then, by \eqref{eeee}, $f(P_1,P_2)=b=f(P_{1},P_{2}^{\star })$ and $%
f(P_{1}^{\prime },P_{2}^{\star })=c.$ Since  $bP_{1}c,$ \eqref{new3} holds and $f$ is regret-free truth-telling.

\end{enumerate}

\noindent $(\Longrightarrow)$ Assume that $f$ is regret-free truth-telling, efficient, and anonymous. We will  prove that $%
f$ is a successive elimination rule or an $A$-maxmin$^{\star }$ rule. There are two cases to consider:

\begin{enumerate}

\item[$\boldsymbol{1}.$] \textbf{there exist $\boldsymbol{a \in X}$ and $\boldsymbol{\overline{P} \in \mathcal{P}^2$ such that $f(\overline{P}%
)=t_{1}(\overline{P}_{i})=a$ for some $i\in \{1,2\}}$}. By efficiency, $f(\overline{P})=t(\overline{P}_{j})=a$ for $j=N\setminus{\{i\}}.$  It follows, by
regret-free truth-telling, that
\begin{equation*}
f(P_i,\overline{P}_{j})=a\text{ for each }P_{i}  \in \mathcal{P}.
\end{equation*}
Then, again by regret-free  truth-telling,
\begin{equation*}
f(P)=a\text{ for each }P \in \mathcal{P}^2 \text{ such that }t(P_{j})=a.
\end{equation*}
Therefore, by anonymity,
\begin{equation*}
f(P)=a\text{ for all }P\text{ such that }a\in \{t(P_{1}),t(P_{2})\}.
\end{equation*}
This implies, by regret-free  truth-telling, that
\begin{equation}
f(P)R_{i}a\text{ for each }P \in \mathcal{P}^2 \text{ and  each }i \in \{1,2\}.  \label{11}
\end{equation}
Let $\widehat{P} \in \mathcal{P}^2$ be such that $\widehat{P}_{1}: b,c,a$ and $\widehat{P}_{2} : c,b,a$. By efficiency, $f(\widehat{P}%
)\in \{c,b\}.$ W.l.o.g., assume that
\begin{equation}\label{hatp}
f(\widehat{P})=b.
\end{equation}
Let $f^\succ$ be the  successive elimination rule with  associated  order $a\succ
b\succ c$ and let $P \in \mathcal{P}^2$. We will prove that $f=f^\succ.$ There are two cases to consider:

\begin{enumerate}

\item[$\boldsymbol{1.1}.$] \textbf{there exists $\boldsymbol{i \in \{1,2\}$ such that $aP_ib$ or   $aP_ic}$}. Therefore, by (\ref{11}),
efficiency, and the definition of $f^\succ$, $f(P)=f^\succ(P).$

\item[$\boldsymbol{1.2}.$] \textbf{$\boldsymbol{bP_ia \text{ and }cP_ia$ for each $i \in \{1,2\}}$}. 
If $t(P_{1})= t(P_{2}),$ then by efficiency $f(P)=t(P_1)=f^\succ(P).$ Assume now that $t(P_1)\neq t(P_2).$ Then, by anonymity and \eqref{hatp}, $f(P)=f(\widehat{P})=b=f^\succ(P).$

\end{enumerate}

\item[$\boldsymbol{2}.$] \textbf{$\boldsymbol{f(P)\neq t_{1}(P_{i})$ for each $P \in \mathcal{P}^2$ and each
$i\in \{1,2\}}$}. First,  let $\widehat{P} \in \mathcal{P}^2$ be such that  $\widehat{P}_{1}:b,c,a$ and $\widehat{P}_{2}:c,b,a$. By efficiency, $f(\widehat{P}%
)\in \{b,c\}.$ Assume, w.l.o.g., that $f(\widehat{P})=b.$ Second, let $\overline{P} \in \mathcal{P}^2$ be such that   $\overline{P}_{1}: b,a,c$ and $\overline{P}_{2}: a,b,c$. By efficiency, $f(\overline{P}%
)\in \{b,a\}.$ Assume, w.l.o.g., that $f(\widehat{P})=a.$ Third, let $\widetilde{P} \in \mathcal{P}^2$ be such that  $\widetilde{P}_{1}: c,a,b$ and $\widetilde{P}_{2}: a,c,b$. By efficiency,  $f(%
\widetilde{P})\in \{c,a\}.$ Assume, w.l.o.g., that  $f(\widetilde{P})=c.$
 We will prove that $f$ is a $A$-maxmin rule$^{\star }$ with associated binary relation  $\succ
^{\star }$ where $b\succ ^{\star }c$, $a\succ ^{\star }b$, and $c\succ ^{\star }a$
. This is, we need to show that 
\begin{equation}
f(P)=\underset{\succ ^{\star }}{\max } \ \mathcal{M}(P) \label{igual}
\end{equation}
for each $P \in \mathcal{P}^2.$ To do so, let $P \in \mathcal{P}^2.$ If  $P\in \{\widehat{P},\overline{P},\widetilde{P}\}$, it is clear that \eqref{igual} holds. Assume $P\in \mathcal{P} \setminus \{\widehat{P},\overline{P},\widetilde{P}\}.$ There are three cases to consider:

\begin{enumerate}

\item[$\boldsymbol{2.1}.$] \textbf{$\boldsymbol{t(P_{1})=t(P_{2})}$}. By efficiency, $f(P)=t(P_1)=\underset{\succ ^{\star }}{\max } \ \mathcal{M}(P),$ so \eqref{igual} holds.

\item[$\boldsymbol{2.2}.$] \textbf{$\boldsymbol{t_{1}(P_{1})\neq t_{1}(P_{2})}$}. As $\left\vert
X\right\vert =3,$ there is $x\in X$ such that $\{x\}=X \setminus %
\{t_{3}(P_{1}),t_{3}(P_{2})\}.$ Therefore, 
as $f(P)\neq t_{1}(P_{i})$ for each $i\in \{1,2\}$ (see hypothesis of Case 2)$%
,$ $f(P)=x.$ Furthermore, as $t_{3}(P_{1})\neq t_{3}(P_{2}),$ $\mathcal{M}%
(P)=\{x\}$ and then, \eqref{igual} holds.

\item[$\boldsymbol{2.3}.$] \textbf{$\boldsymbol{t(P_{1})\neq t(P_{2})$ and $t_{1}(P_{1})=t_{1}(P_{2})}$}.  Then, $(P_{1},P_{2})=(P_{1}^{\prime },P_{2}^{\prime })$ with $P^{\prime
}\in \{\widehat{P},\overline{P},\widetilde{P}\}.$ By anonymity and the fact that \eqref{igual} holds for $P'$,
\begin{equation*}
f(P)=f(P^{\prime })=\underset{\succ ^{\star }}{\max } \ \mathcal{M}(P^{\prime
})=\underset{\succ ^{\star }}{\max } \ \mathcal{M}(P).
\end{equation*} 

\end{enumerate}

\end{enumerate}
Therefore, $f$ is a successive elimination rule or an $A$-maxmin$^{\star }$ rule, as stated. \hfill $\square$

\end{document}